\documentclass[11pt]{article}

\makeatletter
\renewcommand*\@fnsymbol[1]{\the#1}
\makeatother

\usepackage{amsmath}
\usepackage{amsfonts}
\usepackage{amssymb}
\usepackage{amsthm}
\usepackage[english]{babel}
\usepackage{latexsym}
\usepackage[hmargin=2cm,vmargin=4cm]{geometry}
\usepackage{enumerate}
\usepackage{nicefrac}
\usepackage{comment}
\usepackage{mathrsfs}
\usepackage[affil-it]{authblk}

\theoremstyle{plain}
\newtheorem{theorem}{Theorem}[section]
\newtheorem{lemma}[theorem]{Lemma}
\newtheorem{proposition}[theorem]{Proposition}
\newtheorem{corollary}[theorem]{Corollary}
\theoremstyle{definition}
\newtheorem{definition}[theorem]{Definition}
\theoremstyle{remark}

\newtheorem{remark}[theorem]{Remark}
\newtheorem{example}[theorem]{Example}

\numberwithin{equation}{section}

\newcommand{\elig}{\mathscr{S}}
\newcommand{\extens}{\mathscr{E}_\pi(\elig)}
\newcommand{\Dom}{\mathop{\rm dom}\nolimits}

\newcommand{\Closure}{\mathop{\rm cl}\nolimits}

\newcommand{\Interior}{\mathop{\rm int}\nolimits}
\newcommand{\Core}{\mathop{\rm core}\nolimits}

\newcommand{\e}{\varepsilon}
\newcommand{\VaR}{\mathop {\rm VaR}\nolimits}

\newcommand{\HS}{{\mathcal {H}^+}}

\newcommand{\E}{{\mathbb E}}

\newcommand{\N}{{\mathbb N}}

\newcommand{\R}{{\mathbb R}}

\newcommand{\probp}{\mathbb P}
\newcommand{\probq}{\mathbb Q}
\newcommand{\expected}{\mathbb E}

\newcommand{\cA}{{\mathscr{A}}}
\newcommand{\cB}{{\mathscr{B}}}
\newcommand{\cC}{{\mathscr{C}}}
\newcommand{\cD}{{\mathscr{D}}}

\newcommand{\cF}{{\mathcal{F}}}

\newcommand{\cM}{{\mathscr{M}}}

\newcommand{\cP}{{\mathcal{P}}}

\newcommand{\cS}{{\mathscr{S}}}

\newcommand{\cU}{{\mathscr{U}}}

\newcommand{\cX}{{\mathscr{X}}}

\def\keywords{\vspace{.5em}
{\noindent\textbf{Keywords}:\,\relax%
}}

\def\JELclassification{\vspace{.5em}
{\noindent\textbf{JEL classification}:\,\relax%
}}

\def\MSCclassification{\vspace{.5em}
{\noindent\textbf{MSC}:\,\relax%
}}

\makeatletter
\def\@fnsymbol#1{\ensuremath{\ifcase#1\or *\or 1\or 2\or
   3\or 4\or 5\or 6\or 7\or 8\else\@ctrerr\fi}}
\makeatother

\begin{document}

\title{Measuring risk with multiple eligible assets\footnote{Partial support through the SNF project 51NF40-144611 ``Capital adequacy, valuation, and portfolio selection for insurance companies'' is gratefully acknowledged.}}

\author{\sc{Walter Farkas}\thanks{Email: \texttt{walter.farkas@bf.uzh.ch}}\,, \sc{Pablo Koch-Medina}\thanks{Email: \texttt{pablo.koch@bf.uzh.ch}}}
\affil{Department of Banking and Finance, University of Zurich, Switzerland}

\author{\sc{Cosimo Munari}\,\thanks{Email: \texttt{cosimo.munari@math.ethz.ch}}}
\affil{Department of Mathematics, ETH Zurich, Switzerland}

\date{March 3, 2014}

\maketitle

\begin{abstract}
The risk of financial positions is measured by the minimum amount of capital to raise and invest in \textit{eligible} portfolios of traded assets in order to meet a prescribed \textit{acceptability} constraint. We investigate nondegeneracy, finiteness and continuity properties of these risk measures with respect to multiple eligible assets. Our finiteness and continuity results highlight the interplay between the acceptance set and the class of eligible portfolios. We present a simple, alternative approach to the dual representation of convex risk measures by directly applying to the acceptance set the external characterization of closed, convex sets. We prove that risk measures are nondegenerate if and only if the pricing functional admits a positive extension which is a supporting functional for the underlying acceptance set, and provide a characterization of when such extensions exist. Finally, we discuss applications to set-valued risk measures, superhedging with shortfall risk, and optimal risk sharing.
\end{abstract}


\MSCclassification{91B30, 46A40, 46A20, 46A22}

\JELclassification{C60, G11, G22}

\parindent 0em \noindent

\section{Introduction}

The objective of this paper is to provide a variety of new results for risk measures with respect to multiple eligible assets as introduced by Scandolo in~\cite{Scandolo2004} and Frittelli and Scandolo in~\cite{FrittelliScandolo2006}, and further investigated by Artzner, Delbaen, and Koch-Medina in~\cite{ArtznerDelbaenKoch2009}.

\medskip

To understand our financial motivation consider an economy with dates $t=0$ and $t=T$, and let $\cX$ be a topological vector space ordered by a convex cone $\cX_+$ and representing the space of \textit{capital positions} of a given financial institution at time $t=T$. The institution is deemed to be adequately capitalized if its position belongs to a pre-specified {\em acceptance set} $\cA\subset\cX$, i.e. a proper set satisfying $\cA+\cX_+\subset\cA$. The institution has to deal with the following capital adequacy problem: {\em Assuming a position is not acceptable, can it be made acceptable through appropriate management actions and, if so, at which cost?} To turn this question into a mathematically tractable problem we need to specify a nonempty set $\cM$ representing the ``eligible'' or ``admissible'' management actions, a ``cost'' function $c:\cM\to\R$ assigning to each eligible management action its cost, and an ``impact'' function $I:\cX\times\cM\to\cX$ assigning to each position $X$ and each admissible action $m$ the new position $I(X,m)$ obtained from $X$ after implementing $m$. The ``minimal'' cost of making $X$ acceptable can be then expressed as
\begin{equation}
\rho(X): = \inf\{\, c(m): \, m\in\cM, \;\, I(X,m)\in\cA \,\}\,.
\end{equation}

The \textit{general capital requirements} defined in Scandolo~\cite{Scandolo2004} and studied in Frittelli and Scandolo~\cite{FrittelliScandolo2006} and Artzner, Delbaen, and Koch-Medina~\cite{ArtznerDelbaenKoch2009} have this form. The corresponding class of eligible actions consists in a subset $\elig$ of $\cX$, interpreted as the set of payoffs of eligible strategies, equipped with a pricing functional $\pi:\elig\to\R$. For $X\in\cX$ and $Z\in\elig$ the impact function has the form $I(X,Z):=X+Z$, leading to
\begin{equation}
\rho_{\cA,\elig,\pi}(X): = \inf\{\,\pi(Z): \, Z\in\elig, \;\, X+Z\in\cA\,\}\,.
\end{equation}

\medskip

In this paper we will refer to the triple $(\cA,\elig,\pi)$ as a {\em risk measurement regime} and to the map $\rho_{\cA,\elig,\pi}:\cX\to\overline{\R}$ as a {\em risk measure with respect to multiple eligible assets}, or {\em multi-asset risk measure} for short. Clearly, this map can be seen as a generalized version of the risk measures introduced in Artzner, Delbaen, Eber, and Heath~\cite{ArtznerDelbaenEberHeath1999} for which the only allowed management action is to raise capital and invest it in a {\em single} eligible asset, typically taken to be cash.

\medskip

Under the assumption that $\elig$ is a linear space and $\pi:\elig\to\R$ is a linear functional, we provide several new results on the finiteness and continuity properties of multi-asset risk measures. In the convex case, continuity properties are particularly important since they ensure the availability of dual representations. With regard to dual representations, we follow an alternative approach based directly on the geometry of acceptance sets rather than on conjugate functions. Finally, we apply our results to obtain dual representations in several relevant situations, including a dual representation for scalarized set-valued risk measures, for the cost of superhedging with shortfall risk, and for infimal convolutions in the context of risk sharing. Incidentally, note that the generality of the possible underlying spaces, together with the link to set-valued risk measures, shows that our approach can be easily applied to the measurement of the risk of multivariate positions.

\medskip

\textit{Finiteness and continuity}. A natural first step when studying multi-asset risk measures of the form $\rho_{\cA,\elig,\pi}$ is to determine conditions under which they are finite-valued and continuous. This is undertaken in Section~\ref{finiteness continuity section}, where we start by discussing general acceptance sets in Proposition~\ref{finiteness and NAA when a portfolio is order unit} and then narrow down our analysis to convex and coherent acceptance sets in Proposition~\ref{convex strictly positive} and Proposition~\ref{finiteness and coherent core}, respectively. In this respect, a key tool is the simple but far-reaching Lemma~\ref{reduction lemma}, the so-called Reduction Lemma. This result shows how to express a multi-asset risk measure as a risk measure with respect to a single asset by properly enlarging the acceptance set. The ``augmented'' acceptance set consists of all positions which are acceptable up to a zero-cost eligible strategy.

\medskip

\textit{Dual representations}. In Theorem~\ref{pointwise repr formula for multiple assets} we provide a dual representation for convex, lower semicontinuous multi-asset risk measures. Instead of using conjugate duality methods, we exploit the particular structure of risk measures and opt for an alternative approach based on a direct analysis of the geometry of the underlying acceptance sets. This approach proves to be very effective when characterizing nondegeneracy or finiteness of risk measures. In particular, it leads naturally to conditions reflecting the interplay of the acceptance set and the eligible space, which are the two defining objects for multi-asset risk measures. In this respect, a key requirement for nondegeneracy is the existence of extensions of the pricing functional that are supporting functionals for the underlying acceptance set. In Theorem~\ref{ext of pi} we provide equivalent conditions for the existence of such extensions. This result is of independent interest and generalizes classical results on the extension of positive linear functionals.

\medskip

{\em Set-valued risk measures}. In Section~\ref{set valued section} we investigate the link between multi-asset risk measures and set-valued risk measures as defined by Hamel, Heyde and Rudloff in~\cite{HamelHeydeRudloff2011}. Set-valued risk measures have been recently study in connection with capital adequacy or margin requirement problems in market models characterized by proportional transaction costs. In Proposition~\ref{from set-valued to multi} we show that the scalarization of a set-valued risk measure by means of a consistent pricing system can be expressed as a multi-asset risk measure. Hence, we can apply our general results and provide new finiteness and continuity results, as well as dual representations, for convex scalarized set-valued risk measures in Proposition~\ref{dual repr for scal set valued}. Finally, we focus on special acceptance sets satisfying the market compatibility condition introduced in~\cite{HamelHeydeRudloff2011} and we specify the dual representation to this case in Corollary~\ref{dual repr scal compatible set valued}.

\medskip

{\em Superhedging with shortfall risk}. In Section~\ref{shortfall section} we focus on the superhedging problem considered by Arai in~\cite{Arai2011}. We extend to arbitrary (convex) classes of admissible trading strategies the key Lemma 5.1 in~\cite{Arai2011} which was established for so-called $W$-admissible strategies. As a result, we are able to provide in Proposition~\ref{general dual repr arai} a sharper dual representation for the corresponding superhedging cost.

\medskip

\textit{Optimal risk sharing}. In Section~\ref{inf conv section} we show that the infimal convolution of single-asset risk measures can always be expressed as a multi-asset risk measure. Since infimal convolutions are important tools when studying optimal risk sharing across several business lines, our results find a natural application in this context. In particular, we provide a dual representation for the infimal convolution of convex risk measures in Proposition~\ref{dual repr inf conv of risk measures}.

%%%%%%%%%%%%%%%%%%%%%%%%%%%%%%%%%%%%%%%

\subsubsection*{Related literature}
Risk measures of the form $\rho_{\cA,\elig,\pi}$ seem to have been first considered by Scandolo in~\cite{Scandolo2004}. That paper is devoted to providing a financial motivation for that type of risk measures and only contains basic mathematical results. The investigation of multi-asset risk measures was resumed by Frittelli and Scandolo in~\cite{FrittelliScandolo2006}. The main focus there is on developing a theory of risk measures for processes of bounded random variables. However, when restricted to a one-period economy, their risk measures are not genuinely multi-asset since, as remarked after Definition~4.3 in~\cite{FrittelliScandolo2006}, they turn out to coincide with standard cash-additive risk measures on $L^\infty$. Hence, none of our results on finiteness and (lower semi)continuity is contained in that paper. In the convex case, the authors provide a dual representation for finite-valued risk measures by means of conjugate duality. As already mentioned, in this paper we opt for an alternative approach exploiting the particular structure of risk measures and building on a dual representation of the underlying acceptance set. This approach has the advantage to immediately produce conditions highlighting the interplay between the acceptance set (through its barrier cone) and the eligible space (through the set of positive extensions of the pricing functional).

\smallskip

The investigation of multi-asset risk measures was further pursued by Artzner, Delbaen, and Koch-Medina in~\cite{ArtznerDelbaenKoch2009}. The setting there was that of coherent acceptance sets on the space of random variables defined on a finite sample space, and the main focus was on the compatibility of risk measures defined in different currencies. In that context, finiteness and continuity, as well as dual representations, were not investigated. Finally, we mention the paper by Kountzakis~\cite{Kountzakis2009}. There the author proves a representation theorem for coherent multi-asset risk measures under the assumption that $\cX$ is a reflexive Banach space and that some eligible payoff in $\elig$ is an interior point of the positive cone of $\cX$. In particular, this result does not apply to spaces ordered by cones with empty interior, such as $L^p$ spaces, $1\leq p<\infty$, when the underlying probability space is nonatomic.

%%%%%%%%%%%%%%%%%%%%%%%%%%%%%%%%%%%%%%%%%%%%%%%%%

\section{Risk measures in a multi-asset setting}

In this introductory section we investigate the basic properties of multi-asset risk measures defined on an ordered topological vector space.

%%%%%%%%%%%%%%%%%%%%%%%%%%%%%%%%%%%%%%%%%%%%%%%

\subsection{Financial positions and acceptance sets}

The space $\cX$ of financial positions is assumed to be an \textit{ordered topological vector space} over $\R$ with positive cone $\cX_+$. In particular, $\cX_+$ is assumed to be a pointed, convex cone. We use the standard notation $Y\geq X$ whenever $Y-X\in\cX_+$. Unless explicitly stated, we do not assume that $\cX_+$ has nonempty interior. By $\cX'$ we denote the topological dual space of $\cX$, which can itself be viewed as an ordered topological vector space whose positive cone $\cX'_+$ consists of all {\em positive} functionals, i.e. of all elements $\psi\in\cX'$ such that $\psi(X)\geq0$ whenever $X\in\cX_+$. We use the following standard terminology suggested by the natural duality pairing between $\cX$ and $\cX^\prime$: An element $X\in\cX_+$ is said to be a {\em strictly positive element} whenever $\psi(X)>0$ for every nonzero $\psi\in\cX^\prime_+$ and a functional $\psi\in\cX'_+$ is said to be a {\em strictly positive functional} whenever $\psi(X)>0$ for every nonzero $X\in\cX_+$.

\smallskip

If $\cA$ is a subset of $\cX$, we denote by $\Interior(\cA)$, $\Closure(\cA)$ and $\partial\cA$ the interior, the closure and the boundary of~$\cA$, respectively. Moreover, by $\Core(\cA)$ we denote the \textit{core} of $\cA$, i.e. the set of all positions $X\in\cA$ such that for each $Y\in\cX$ there exists $\e>0$ with $X+\lambda Y\in\cA$ whenever $\left|\lambda\right|<\e$. Note that $\cA$ is said to be a \textit{cone} if $\lambda\cA\subset\cA$ for every $\lambda\geq0$.

\begin{remark}
\label{remark: order structure standard spaces}
Let $(\Omega,\cF,\probp)$ be probability space and $d\in\N$. The ordered topological vector spaces used in financial mathematics are typically subspaces of the space $L^0_d$ of all $\cF$-measurable functions $X:\Omega\to\R^d$. As usual, elements in $L^0_d$ that coincide almost surely are identified. This space is naturally ordered by the convex cone of all $X\in L^0_d$ such that $X\in\R^d_+$ almost surely, where $\R^d_+$ denotes the standard positive cone in $\R^d$. Standard subspaces of $L^0_d$ we consider are $d$-dimensional $L^p$ spaces, $1\leq p\le\infty$, and Orlicz spaces. These spaces are equipped with the usual norms unless explicitly stated. If $d=1$ we will drop the subscript $d$.
\end{remark}

We start by recalling the notion of an acceptance set.

\begin{definition}
\label{acceptance set definition}
A set $\cA\subset\cX$ is called an {\em acceptance set} whenever the following two conditions are satisfied:
\begin{enumerate}[(i)]
\item $\cA$ is a nonempty, proper subset of $\cX$ (non-triviality);\label{non-triviality axiom}
\item if $X\in\cA$ and $Y\geq X$ then $Y\in\cA$ (monotonicity).\label{monotonicity axiom}
\end{enumerate}
\end{definition}

Later on we will focus on convex acceptance sets and coherent acceptance sets, i.e. acceptance sets that are convex cones. We refer to~\cite{FoellmerSchied2011} for a financial interpretation and examples.

\medskip

We recall a simple but fundamental property of acceptance sets, see Lemma~3.11 in~\cite{FarkasKochMunari2012a}: Any positive halfspace containing an acceptance set must be generated by a positive functional. Here, the {\em (positive) halfspace} generated by a functional $\psi:\cX\to\R$ at a level $\alpha\in\R$ is the set
\begin{equation}
\HS(\psi,\alpha):=\{\,X\in\cX: \, \psi(X)\ge\alpha\,\}\,.
\end{equation}

\begin{lemma}
\label{halfspaces containing acceptance sets}
Let $\cA\subset\cX$ be a nonempty monotone set, and take a linear functional $\psi:\cX\to\R$. If $\cA\subset\HS(\psi,\alpha)$ for some $\alpha\in\R$, then $\psi$ is positive.
\end{lemma}

The above result will be used repeatedly without further mention.

%%%%%%%%%%%%%%%%%%%%%%%%%%%%%%%%%%%%%%%%%%%%%%%%%

\subsection{Risk measures with multiple eligible assets}

Consider a financial market described by a vector subspace $\cM\subset\cX$. The space $\cM$ is called the \textit{marketed space} and its elements represent the {\em marketed payoffs}. We think of marketed payoffs as the payoffs that can be replicated by executing ``admissible'' trading strategies. We assume the existence of a linear pricing functional $\pi:\cM\to\R$ such that, for every market payoff $Z\in\cM$, the quantity $\pi(Z)$ represents the initial value of the replicating strategy\footnote{For our purposes it suffices to know that there is a linear pricing functional assigning market values to marketed payoffs. For more details on the underlying market models, also for the case of infinite dimensional marketed spaces, we refer to Clark~\cite{Clark1993} and Kreps~\cite{Kreps1981}.}. Moreover, we assume that the market is free of arbitrage opportunities by requiring that $\pi$ is a strictly positive functional, i.e. that $\pi(Z)>0$ for any nonzero positive $Z\in\cM$.

\medskip

If the financial position of a financial institution is not acceptable with respect to a given acceptance set $\cA\subset\cX$, it is natural to ask which management actions can turn it into an acceptable position and at which cost. In this paper we allow financial institutions to modify the acceptability profile of their capital position by raising capital and investing it in admissible portfolios with payoff in a fixed subspace $\elig$ of the marketed space $\cM$. The restriction to strategies with payoffs in $\elig\subset\cM$ is meant to capture situations where a financial institution may only be able to operate in a segment of the full market. The space $\elig$ will be referred to as the {\em eligible space}.

\medskip

\begin{remark}
Note that our market model is general enough to accommodate a variety of common situations. In particular, the eligible space $\elig$ may be finite or infinite dimensional. In the former case, we can always assume that $\elig$ is spanned by a finite number of linearly independent payoffs. Typical examples of infinite-dimensional marketed spaces arise in connection to continuous-time trading between time $t=0$ and $t=T$.
\end{remark}

\medskip

Throughout the paper we will use the following notation. The extended real line $\R\cup\{\pm\infty\}$ will be denoted by $\overline{\R}$. For a subspace $\elig\subset\cM$ and~$m\in\R$, we will write
\begin{equation}
\elig_m:=\{\,Z\in\elig: \, \pi(Z)=m\,\}\,.
\end{equation}

For the triple $(\cA,\cS,\pi)$, where $\cA\subset\cX$ is an arbitrary subset, $\elig$ is a subspace of $\cM$, and $\pi:\cM\to\R$ is the pricing functional, we define the map $\rho_{\cA,\elig,\pi}:\cX\to\overline{\R}$ by setting
\begin{equation}
\rho_{\cA,\elig,\pi}(X):=\inf\{\,\pi(Z): \, Z\in\elig, \;\, X+Z\in\cA\,\}\,.
\end{equation}

We say $\rho_{\cA,\cS,\pi}$ is a risk measure with respect to {\em multiple eligible assets} whenever $\dim(\cS)\ge 2$ and a risk measure with respect to a {\em single eligible asset} if $\dim(\cS)=1$. In the latter case $\elig$ is generated by a nonzero $U\in\cM$ and, when $\pi(U)\neq0$, we set
\begin{equation}
\label{def single asset}
\rho_{\cA,U,\pi}(X):=\rho_{\cA,\elig,\pi}(X)=\inf\Big\{\,m\in\R: \, X+\frac{m}{\pi(U)}U\in\cA\,\Big\}\,.
\end{equation}

The interpretation of the map $\rho_{\cA,\elig,\pi}$ is clear. If $\cA\subset\cX$ is an acceptance set and $\elig$ a subspace of $\cM$, then for a given $X\in\cX$ the quantity $\rho_{\cA,\elig,\pi}(X)$ represents, when finite, the ``minimum'' amount of capital that needs to be raised and invested at time $t=0$ in portfolios with payoff in $\elig$ in order to make the position $X$ acceptable.

\medskip

\begin{remark}
Risk measures of the form~\eqref{def single asset} were originally introduced in~\cite{ArtznerDelbaenEberHeath1999} for financial positions on a finite sample space. Most of the subsequent literature focused on {\em cash-additive} risk measures, i.e. risk measures for which the eligible asset is cash. The case of a general eligible asset has been recently investigated by the authors in~\cite{FarkasKochMunari2012a} for bounded financial positions, and in~\cite{FarkasKochMunari2013b} for positions belonging to a general ordered topological vector space.
\end{remark}

\medskip

For $\rho_{\cA,\elig,\pi}$ to be a sound risk measurement tool, we need to impose additional conditions on the triple $(\cA,\elig,\pi)$.

\medskip

\begin{definition}
The triple $(\cA,\elig,\pi)$ is called a \textit{risk measurement regime} if $\cA$ is an acceptance set and $\elig$ a subspace of $\cM$ satisfying $\elig\cap\cX_+\neq\{0\}$.

We say that $(\cA,\elig,\pi)$ is a convex, respectively coherent, risk measurement regime when $\cA$ is a convex, respectively coherent, acceptance set.
\end{definition}

\medskip

\begin{remark}
Let $(\cA,\elig,\pi)$ be a risk measurement regime. Requiring that $\elig$ contains some nonzero positive payoff $U$ is a natural assumption which is satisfied whenever $U$ is the payoff of a bond or a stock. In this case, since $\pi(U)>0$ by no arbitrage, we can use the payoff $U$ as an effective vehicle to decrease the amount of required capital given that $\rho_{\cA,\elig,\pi}(X+\lambda U)=\rho_{\cA,\elig,\pi}(X)-\lambda\pi(U)<\rho_{\cA,\elig,\pi}(X)$ holds for any $\lambda>0$.
\end{remark}

\medskip

Before stating the basic properties of multi-asset risk measures, we recall some standard terminology. The \textit{effective domain} of a map $\rho:\cX\to\overline{\R}$ is
\begin{equation}
\Dom(\rho):=\{\,X\in\cX: \, \rho(X)<\infty\,\}\,.
\end{equation}
The function $\rho$ is said to be \textit{convex}, \textit{subadditive}, or \textit{positively homogeneous}, whenever its epigraph is convex, closed under addition, or it is a cone, respectively. The map $\rho$ is said to be {decreasing} if $\rho(X)\geq\rho(Y)$ whenever $X\le Y$. Finally, if $\elig$ is a subspace of $\cM$, we say that $\rho$ is {\em $\elig$-additive} whenever
\begin{equation}
\rho(X+Z)=\rho(X)-\pi(Z) \ \ \ \mbox{for all} \ X\in\cX \ \mbox{and} \ Z\in\elig\,.
\end{equation}

The following lemma records a few straightforward properties of maps of the form $\rho_{\cA,\elig,\pi}$ and is stated without proof.

\medskip

\begin{lemma}
\label{multiple assets translation and monotonicity}
Let $\cA\subset\cX$ be an arbitrary nonempty set and $\elig$ a subspace of $\cM$. Then $\rho_{\cA,\elig,\pi}$ satisfies the following properties:
\begin{enumerate}[(i)]
 	\item $\Dom(\rho_{\cA,\elig,\pi})=\cA+\elig$;
 	\item $\rho_{\cA,\elig,\pi}$ is $\elig$-additive;
  \item if $\cA$ is convex, closed under addition, or a cone, then $\rho_{\cA,\cS,\pi}$ is convex, subadditive, or positively homogeneous, respectively;
  \item if $\cA$ is a monotone set, then $\rho_{\cA,\elig,\pi}$ is decreasing;
  \item if $(\cA,\elig,\pi)$ is a risk measurement regime, the set $\{\,\pi(Z): \, Z\in\elig, \;\, X+Z\in\cA\,\}$ is a (possibly empty) interval which is unbounded to the right for any $X\in\cX$.
\end{enumerate}
\end{lemma}

\medskip

\begin{remark}
Note that $\{\, X\in\cX : \, \rho_{\cA,\elig,\pi}(X)\le 0\, \}$ is a nonempty, monotone set. Proposition~\ref{pointwise semicontinuity} below implies that, if $\rho_{\cA,\elig,\pi}$ is lower semicontinuous, then this set coincides with $\Closure(\cA+\elig_0)$. Hence, in contrast to the single-asset case, we cannot use $\{\, X\in\cX : \, \rho_{\cA,\elig,\pi}(X)\le 0\, \}$ to recover the closure of the original acceptance set $\cA$, but only the closure of the ``augmented'' acceptance set $\cA+\elig_0$.
\end{remark}

%%%%%%%%%%%%%%%%%%%%%%%%%%%%%%%%%%%%%%%%%%%%%%%%%%%%%%%%%%%%%%%%%%%%%%

\subsection{From several assets to a single asset}

In this brief section we extend and complement the Corollary on page 112 in~\cite{ArtznerDelbaenKoch2009} and show how to convert risk measures with multiple eligible assets into risk measures with respect to a single eligible asset by properly augmenting the underlying acceptance set.

\medskip

\begin{lemma}[Reduction Lemma]
\label{reduction lemma}
Let~$\cA$ be an arbitrary nonempty set in $\cX$ and $\elig$ a subspace of $\cM$. If $U\in\elig$ is a nonzero positive payoff, then
\begin{equation}
\rho_{\cA,\elig,\pi}(X)=\rho_{\cA+\elig_0,U,\pi}(X)
\end{equation}
for every $X\in\cX$.
\end{lemma}
\begin{proof}
Note that $\elig_m=\frac{m}{\pi(U)}U+\elig_0$ for any $m\in\R$. It follows that for every $X\in\cX$
\begin{eqnarray}
\rho_{\cA,\elig,\pi}(X)
&=&
\inf\{\,m\in\R: \, (X+\elig_m)\cap\cA\neq\emptyset\,\} \\
&=&
\inf\Big\{\,m\in\R: \, X+\frac{m}{\pi(U)}U\in\cA+\elig_0\,\Big\} = \rho_{\cA+\elig_0,U,\pi}(X)\,,
\end{eqnarray}
concluding the proof.
\end{proof}

\medskip

\begin{remark}
\label{when the sum is acceptance set}
The augmented set $\cA+\elig_0$ is always nonempty and monotone. Hence, it is itself an acceptance set whenever it is strictly contained in $\cX$.
\end{remark}

%%%%%%%%%%%%%%%%%%%%%%%%%%%%%%%%%%%%%%%%%%%%%%%%%%%%%%%%%

\section{No acceptability arbitrage, finiteness, and continuity}
\label{finiteness continuity section}

In this section we provide a variety of finiteness and continuity results for multi-asset risk measures.

\medskip

We start by showing that $\rho_{\cA,\elig,\pi}(0)>-\infty$ if and only if the pricing functional is bounded from below on $\cA\cap\elig$. If this property is not satisfied, then there exist acceptable positions at arbitrarily negative cost allowing for what could be called ``acceptability arbitrage''. As we will see, the absence of ``acceptability arbitrage'' is typically sufficient to ensure finiteness and continuity of multi-asset risk measures.

\medskip

\begin{lemma}
\label{lemma no acc arbitrage}
Let $(\cA,\elig,\pi)$ be a risk measurement regime. The following statements are equivalent:
\begin{enumerate}[(a)]
	\item $\cA\cap\{\,Z\in\elig: \, \pi(Z)\leq m\,\}=\emptyset$ \,for some $m\in\R$;
	\item $\cA\cap\elig_m=\emptyset$ \,for some $m\in\R$;
	\item $\rho_{\cA,\elig,\pi}(0)>-\infty$.
\end{enumerate}
\end{lemma}
\begin{proof}
Clearly, \textit{(a)} implies \textit{(b)}, which is equivalent to \textit{(c)}. Assume \textit{(b)} holds so that $\cA\cap\elig_m=\emptyset$ for a fixed $m\in\R$, but $Z_k\in\cA\cap\elig_k$ for some $k<m$. Take a positive $U\in\elig$ with $\pi(U)=1$. Then $Z_k+(m-k)U\in\cA\cap\elig_m$, contradicting \textit{(b)}. Hence, \textit{(b)} implies \textit{(a)} concluding the proof.
\end{proof}

\medskip

\begin{remark}
\begin{enumerate}[(i)]
	\item The condition $\rho_{\cA,\elig,\pi}(0)>-\infty$ stated in the preceding lemma generalizes the condition of ``no acceptability arbitrage'' introduced in \cite{ArtznerDelbaenKoch2009} in the context of coherent acceptance sets and finite state spaces. In fact, in that setting the two conditions are equivalent, as implied by Proposition~\ref{finiteness and coherent core} below.
	\item In a pricing context, Jaschke and K\"{u}chler in~\cite{JaschkeKuchler2001} introduced the assumption of absence of {\em good deals} (of the first kind) $\cA\cap\elig_0\setminus\{0\}=\emptyset$. Clearly, this condition implies the ``no acceptability arbitrage'' condition \textit{(b)} in the above lemma.
\end{enumerate}
\end{remark}

%%%%%%%%%%%%%%%%%%%%%%%%%%%%%%%%%%%%%%%%%%%%%

\subsection{General acceptance sets}
\label{general acc sets section}

First, we provide a finiteness and continuity result under the assumption that the space $\elig$ contains an order unit. Recall that an element $U\in\cX_+$ is called an \textit{order unit} if $U\in\Core(\cX_+)$. Note that, if $\cX$ admits an order unit, every acceptance set $\cA\subset\cX$ has nonempty core, since $\cA+\cX_+\subset\cA$.

\medskip

\begin{remark}
\label{remark on core for standard spaces}
\begin{enumerate}[(i)]
	\item If the positive cone $\cX_+$ has nonempty interior, then order units are precisely the interior points of $\cX_+$.
	\item Let $(\Omega,\cF)$ be a finite measurable space, and consider the space $\cX$ of all measurable functions $X:\Omega\to\R$. The interior of the positive cone in $\cX$ is nonempty and consists of all $X\in\cX$ such that $X(\omega)>0$ for all $\omega\in\Omega$.
	\item Consider the space $L^\infty$ over a general probability space $(\Omega,\cF,\probp)$. The corresponding positive cone has nonempty interior, consisting of all elements $X$ which are bounded away from zero, i.e. such that $X\geq\e$ almost surely for some $\e>0$.
	\item Let $(\Omega,\cF,\probp)$ be a nonatomic probability space. The positive cone of $L^p$, for $1\leq p<\infty$, has empty core. This is also the case for nontrivial Orlicz hearts $H^\Phi$ with respect to an Orlicz function $\Phi$.
\end{enumerate}
\end{remark}

\medskip

\begin{proposition}
\label{finiteness and NAA when a portfolio is order unit}
Let $(\cA,\elig,\pi)$ be a risk measurement regime. Assume that $\elig$ contains an order unit $U$. Then:
\begin{enumerate}[(i)]
\item $\rho_{\cA,\elig,\pi}(X)<\infty$ \,for every $X\in\cX$.
\item The following statements are equivalent:
\begin{enumerate}[(a)]
  \item $\rho_{\cA,\elig,\pi}(X)>-\infty$ \,for every $X\in\cX$;
	\item $\rho_{\cA,\elig,\pi}(0)>-\infty$;
	\item $\Core(\cA)\cap\elig_m=\emptyset$ \,for some $m\in\R$;
	\item $\cA+\elig_0\neq\cX$.
\end{enumerate}
\end{enumerate}
In particular, $\rho_{\cA,\elig,\pi}$ is finite-valued if and only if $\rho_{\cA,\elig,\pi}(0)>-\infty$. In this case, $\rho_{\cA,\elig,\pi}$ is continuous whenever $U$ is an interior point of $\cX_+$.
\end{proposition}
\begin{proof}
\textit{(i)} Take $X\in\cX$ and $Y\in\cA$. Since $U$ is an order unit, we have that $Y-X\leq \lambda U$ for some $\lambda>0$. As a result we obtain $\rho_{\cA,\elig,\pi}(X)\leq\rho_{\cA,\elig,\pi}(Y)+\lambda\pi(U)\leq\lambda\pi(U)<\infty$.

\smallskip

\textit{(ii)} It is clear that {\it (a)} implies \textit{(b)}, which in turn implies~\textit{(c)}. Assume now that~{\it (c)} holds for $m\in\R$ but $\cA+\elig_0=\cX$. As a result, given $Z_m\in\elig_m$ we can find $Y\in\cA$ and $Z_0\in\elig_0$ with $Z_m-U=Y+Z_0$. But then $Z_m-Z_0=Y+U\in\Core(\cA)\cap\elig_m$, contradicting~\textit{(c)}. Hence,~\textit{(c)} implies~\textit{(d)}. Finally, if~{\em (d)} holds then $\cA+\elig_0$ is an acceptance set so that~{\em (a)} follows by applying the Reduction Lemma and Proposition~3.1 in~\cite{FarkasKochMunari2012a}.

\smallskip

Finally, assume that $\rho_{\cA,\elig,\pi}(0)>-\infty$ and $U$ is an interior point of $\cX_+$. Since $\cA+\elig_0$ is a proper subset of $\cX$, the continuity of $\rho_{\cA,\elig,\pi}$ follows again from Proposition~3.1 in~\cite{FarkasKochMunari2012a}.
\end{proof}

%%%%%%%%%%%%%%%%%%%%%%%%%%%%%%%%%%%%%%%%%%

\subsection{Convex acceptance sets}
\label{convex acc subsection}

In general $L^p$ spaces or Orlicz hearts we cannot apply Proposition~\ref{finiteness and NAA when a portfolio is order unit} to obtain finiteness and continuity results since, by Remark~\ref{remark on core for standard spaces}, the core of the positive cone in these spaces is empty. However, as remarked below, the positive cone in these spaces does possess \textit{strictly positive} elements, i.e. positive elements $U\in\cX_+$ such that $\psi(U)>0$ for every nonzero positive functional $\psi\in\cX'$. This will be enough to provide a sufficient condition for finiteness and continuity in case of convex acceptance sets.

\medskip

\begin{remark}
\label{strictly positive elements for standard spaces}
\begin{enumerate}[(i)]
	\item If the positive cone $\cX_+$ has nonempty interior, the notion of strictly positive elements coincides with that of order units and interior points of $\cX_+$.
	\item Let $(\Omega,\cF,\probp)$ be a probability space. If we equip $L^\infty$ with the weak$^\ast$ topology $\sigma(L^\infty,L^1)$, the set of strictly positive elements is larger than the core of $L^\infty_+$ and consists of all positive $U$ such that $U>0$ almost surely.
	\item Let $(\Omega,\cF,\probp)$ be a probability space. The strictly positive elements in the space $L^p$, $1\leq p<\infty$, are precisely all positive $U$ such that $U>0$ almost surely. The same is true for any nontrivial Orlicz heart $H^\Phi$.
\end{enumerate}
\end{remark}

\medskip

\begin{proposition}
\label{convex strictly positive}
Let $(\cA,\elig,\pi)$ be a convex risk measurement regime such that $\cA$ has nonempty interior. Assume $\elig$ contains a strictly positive element $U$. Then:
\begin{enumerate}[(i)]
 \item $\rho_{\cA,\elig,\pi}(X)<\infty$ \,for every $X\in\cX$.
 \item The following are equivalent:
\begin{enumerate}[(a)]
  \item $\rho_{\cA,\elig,\pi}(X)>-\infty$ \,for every $X\in\cX$;
	\item $\rho_{\cA,\elig,\pi}(0)>-\infty$;
	\item $\Interior(\cA)\cap\elig_m=\emptyset$ \,for some $m\in\R$;
	\item $\cA+\elig_0\neq\cX$.
\end{enumerate}
\end{enumerate}
In particular, $\rho_{\cA,\elig,\pi}$ is finite-valued if and only if $\rho_{\cA,\elig,\pi}(0)>-\infty$. In this case, $\rho_{\cA,\elig,\pi}$ is also continuous.
\end{proposition}
\begin{proof}
{\em (i)} If $\rho_{\cA,\elig,\pi}(X)=\infty$ for some $X\in\cX$, then $(X+\elig)\cap\cA=\emptyset$. By separation we find a nonzero positive $\psi\in\cX'$ such that $\psi(X+Z)\leq\psi(A)$ for all $Z\in\elig$ and $A\in\cA$. In particular, $\psi(X)+\lambda\psi(U)\leq\psi(A)$ holds for every $\lambda\in\R$ and $A\in\cA$. This implies $\psi(U)=0$ which is impossible since $U$ is a strictly positive element. In conclusion, $\rho_{\cA,\elig,\pi}(X)<\infty$ must hold for all $X\in\cX$.

\smallskip

\textit{(ii)} Clearly, \textit{(a)} implies \textit{(b)} which implies \textit{(c)}. Assume $\Interior(\cA)\cap\elig_m=\emptyset$ for some $m\in\R$. By separation there exists a nonzero positive $\psi\in\cX'$ such that $\psi(Z_m)\leq\psi(A)$ for all $Z_m\in\elig_m$ and $A\in\cA$. If \textit{(d)} does not hold, for every $X\in\cX$ we find $A\in\cA$ and $Z_0\in\elig_0$ with $X=A+Z_0$. Since
\begin{equation}
\psi\Big(\frac{m}{\pi(U)}U-Z_0\Big)\leq\psi(A)\,,
\end{equation}
it follows that $\frac{m}{\pi(U)}\psi(U)\leq\psi(X)$. But this cannot hold for every $X\in\cX$, hence we must have $\cA+\elig_0\neq\cX$. Thus, \textit{(c)} implies \textit{(d)}.

\smallskip

Finally, assume \textit{(d)} holds so that $\cA+\elig_0$ is an acceptance set. Since $\cA$ has nonempty interior, the same is true for $\cA+\elig_0$. As a result, we obtain that $\rho_{\cA,\elig,\pi}$ does not attain the value $-\infty$ by combining the Reduction Lemma and Corollary 3.14 in~\cite{FarkasKochMunari2012a}, showing that {\em (a)} holds. Moreover, if $\rho_{\cA,\elig,\pi}(0)>-\infty$ then \textit{(d)} holds and, hence, $\rho_{\cA,\elig,\pi}$ is continuous by the same result.
\end{proof}

%%%%%%%%%%%%%%%%%%%%%%%%%%%%%%%%%%%%

\subsection{Coherent acceptance sets}

For coherent acceptance sets we can obtain finiteness and continuity results under even weaker conditions on the space $\cS$ than those of Proposition~\ref{convex strictly positive}. In this case, the key assumption is that $\elig$ contains a positive payoff belonging to the core of $\cA$.

\medskip

\begin{remark}
The condition that $\elig$ contains a positive payoff belonging to $\Core(\cA)$ is indeed weaker than requiring that $\cS$ contains a strictly positive element. For example, let $(\Omega,\cF,\probp)$ be a probability space and fix $1\leq p\leq\infty$. Consider the coherent acceptance set based on {\em Tail Value-at-Risk} at level $\alpha\in(0,1)$
\begin{equation}
\cA^\alpha:=\bigg\{\,X\in L^p: \, \frac{1}{\alpha}\int^{\alpha}_{0}\VaR_\beta(X)d\beta\leq0\,\bigg\}\,,
\end{equation}
where
\begin{equation}
\VaR_\beta(X):=\inf\{\,m\in\R: \, \probp[X+m<0]\leq\beta\,\}
\end{equation}
is the {\em Value-at-Risk} of $X$ at level $\beta$. As shown by Lemma 4.3 in~\cite{FarkasKochMunari2012a}, a positive $U\in L^p$ belongs to $\Core(\cA^\alpha)$ if and only if $\probp[U=0]<\alpha$. In particular, $U$ need not be an order unit or a strictly positive element.
\end{remark}

\medskip

\begin{proposition}
\label{finiteness and coherent core}
Let $(\cA,\elig,\pi)$ be a coherent risk measurement regime. Assume there exists a nonzero positive $U\in\elig$ such that $U\in\Core(\cA)$. Then:
\begin{enumerate}[(i)]
 \item $\rho_{\cA,\elig,\pi}(X)<\infty$ \,for every $X\in\cX$.
 \item The following statements are equivalent:
\begin{enumerate}[(a)]
  \item $\rho_{\cA,\elig,\pi}(X)>-\infty$ \,for every $X\in\cX$;
	\item $\rho_{\cA,\elig,\pi}(0)>-\infty$;
	\item $\Core(\cA)\cap\elig_0=\emptyset$;
	\item $\cA+\elig_0\neq\cX$.
\end{enumerate}
\end{enumerate}
In particular, $\rho_{\cA,\elig,\pi}$ is finite-valued if and only if $\rho_{\cA,\elig,\pi}(0)>-\infty$. In this case, if $U\in\Interior(\cA)$ then $\rho_{\cA,\elig,\pi}$ is continuous.
\end{proposition}
\begin{proof}
{\em (i)} Since $U\in\Core(\cA)$, for every $X\in\cX$ we find $\lambda>0$ such that $U+\lambda X\in\cA$ or, equivalently, $X+\frac{1}{\lambda}U\in\cA$. Hence, $\rho_{\cA,\elig,\pi}(X)<\infty$.

\smallskip

{\em (ii)} It is clear that {\it (a)} implies \textit{(b)}. Assume that \textit{(c)} does not hold, so that there exists $Z_0\in\Core(\cA)\cap\elig_0$. Take $m\in\R$. For any $Z_m\in\elig_m$ we can find $\lambda>0$ such that $Z_0+\lambda Z_m\in\cA$, hence $Z_m+\frac{1}{\lambda}Z_0\in\cA\cap\elig_m$. As a result, $\rho_{\cA,\elig,\pi}(0)=-\infty$. Hence, it follows that \textit{(b)} implies \textit{(c)}.

\smallskip

Now assume \textit{(c)} holds. By the algebraic separation in Theorem 5.61 in~\cite{AliprantisBorder2006}, we find a nonzero positive linear functional $\psi:\cX\to\R$ such that $\psi(Z_0)\leq\psi(A)$ or, equivalently, $\psi(A +Z_0)\geq0$ for all $Z_0\in\elig_0$ and $A\in\cA$. As a result, we must have $\cA+\elig_0\neq\cX$ since otherwise $\psi(X)\geq0$ would hold for every $X\in\cX$. It follows that \textit{(d)} holds.

\smallskip

Finally, assume that {\it (d)} holds. Since $\cA+\elig_0$ is then a coherent acceptance set, we can apply Theorem 3.16 in~\cite{FarkasKochMunari2012a} combined with the Reduction Lemma to obtain \textit{(a)}. To prove continuity, assume $\rho_{\cA,\elig,\pi}$ is finite-valued and $U\in\Interior(\cA)$. Then $U\in\Interior(\cA+\elig_0)$ and therefore $\rho_{\cA,\elig,\pi}$ is continuous by the Reduction Lemma and Theorem 3.16 in~\cite{FarkasKochMunari2012a}.
\end{proof}

%%%%%%%%%%%%%%%%%%%%%%%%%%%%%%%%%%%%%%%%%%%%%%%%%%%%%%%%%%%%%%

\subsection{Semicontinuity properties}
\label{semicontinuity subsection}

As is well known, lower semicontinuity plays a key role in the dual representation of convex maps. Before dealing with dual representations of multi-asset risk measures in the next section, we complement the continuity results established above by briefly discussing conditions for a risk measure $\rho_{\cA,\elig,\pi}$ to be lower (or upper) semicontinuous.

\medskip

Recall that a function $\rho:\cX\to\overline{\R}$ is said to be \textit{lower semicontinuous} at a point $X\in\cX$ if for every $\e>0$ there exists a neighborhood $\cU$ of $X$ such that $\rho(Y)\geq\rho(X)-\e$ for all $Y\in\cU$. We say that $\rho$ is \textit{(globally) lower semicontinuous} if it is lower semicontinuous at every $X\in\cX$. The corresponding {\em upper semicontinuity} properties for $\rho$ are obtained by requiring lower semicontinuity of $-\rho$.

\medskip

The following characterization of semicontinuity easily follows from Lemma~2.5 in~\cite{FarkasKochMunari2012a} by means of the Reduction Lemma.

\medskip

\begin{proposition}
\label{pointwise semicontinuity}
Let $(\cA,\elig,\pi)$ be a risk measurement regime, and take $X\in\cX$. The following statements hold:
\begin{enumerate}[(i)]
\item $\rho_{\cA,\elig,\pi}$ is lower semicontinuous at $X$ if and only if $X+Z\notin\Closure(\cA+\elig_0)$ for any $Z\in\elig$ with $\pi(Z)<\rho_{\cA,\elig,\pi}(X)$;
\item $\rho_{\cA,\elig,\pi}$ is upper semicontinuous at $X$ if and only if $X+Z\in\Interior(\cA+\elig_0)$ for any $Z\in\cM(\elig)$ with $\pi(Z)>\rho_{\cA,\elig,\pi}(X)$.
\end{enumerate}

In particular, if $\cA+\elig_0$ is closed, respectively open, then $\rho_{\cA,\elig,\pi}$ is lower, respectively upper, semicontinuous on $\cX$.
\end{proposition}

\medskip

In case $\elig$ is one dimensional we clearly have $\elig_0=\{0\}$, hence for single-asset risk measures the closedeness of $\cA$ implies lower semicontinuity as a consequence of the result above. Unfortunately, as the following example shows, closedeness of $\cA$ no longer suffices to ensure lower semicontinuity in the multi-asset case.

\medskip

\begin{example}
\label{remark on semicont}
Let $(\Omega,\cF,\probp)$ be an infinite probability space and take $A\subset\Omega$ with $0<\probp(A)<1$. Let $\cX:=L^1$ and $\cA:=L^1_+$ and assume $U\in L^1_+$ is unbounded from above on $A$ and equal to $1$ on $A^c$. Taking $\elig$ to be the vector space spanned by $U$ and $V:=U+1_A$, define a linear functional $\pi:\elig\to\R$ by setting $\pi(U)=\pi(V):=1$. It is easy to see that $\cA+\elig_0$ consists of all positions $X$ that are bounded from below on $A$ and that are nonnegative on $A^c$ so that $\Closure(\cA+\elig_0)=\{\,X\in L^1: \, X1_{A^c}\geq0\,\}$. Taking $X:=-U1_{A}$ one checks that $\rho_{\cA,\elig,\pi}(X)=1$ while $\rho_{\Closure(\cA+\elig_0),U,\pi}(X)=0$, hence $\rho_{\cA,\elig,\pi}$ is not lower semicontinuous at $X$ as a consequence of the Reduction Lemma and Proposition~\ref{pointwise semicontinuity}.
\end{example}

\medskip

The next result establishes a sufficient condition for $\cA+\elig_0$ to be closed in case $\cA$ is also closed. This allows to apply the previous proposition to obtain lower semicontinuity of the corresponding multi-asset risk measure.

\medskip

\begin{proposition}
\label{closedeness of a sum}
Let $\cA\subset\cX$ be a closed acceptance set with $0\in\cA$. Assume that $\cA$ is either a cone or convex and let $\elig$ be a finite dimensional subspace of $\cM$. If $\cA\cap\elig_0=\{0\}$, then $\cA+\elig_0$ is closed.
\end{proposition}
\begin{proof}
If $\elig_0=\{0\}$ the assertion is trivial. Assume $\elig_0$ is generated by a nonzero payoff $Z\in\cM$. Let $(A_\alpha+\lambda_\alpha Z)$ be a net in $\cA+\elig_0$ converging to some $X\in\cX$. If $(\lambda_\alpha)$ is unbounded, then we can find a subnet $(\lambda_\beta)$ diverging either to $\infty$ or to $-\infty$. Without loss of generality we may assume $\left|\lambda_\beta\right|>1$. It follows that $(A_\beta/\lambda_\beta)$ has limit $-Z$, respectively $Z$. Since $A_\beta/\lambda_\beta$ belongs to $\cA$, respectively $-\cA$, by the closedeness of $\cA$ we conclude that $\cA\cap\elig_0$ contains a nonzero element, contradicting the assumption. Therefore $(\lambda_\alpha)$ must be bounded. Passing to a convergent subnet, it is easy to show that the limit $X$ lies in $\cA$, implying that $\cA+\elig_0$ is closed. We can conclude by induction on the dimension of $\elig_0$.
\end{proof}

\medskip

\begin{remark}
For a convex set $\cA$, the statement in the above proposition also follows from the well-known theorem by Dieudonn\'{e}, see Theorem 1.1.8 in \cite{Zalinescu2002}. Note that we do not require convexity in case $\cA$ is a cone. Note also that the condition $\cA\cap\elig_0=\{0\}$ is equivalent to the absence of good deals (of the first kind) introduced by Jaschke and K\"{u}chler~\cite{JaschkeKuchler2001} in a pricing context.
\end{remark}

%%%%%%%%%%%%%%%%%%%%%%%%%%%%%%%%%%%%%%%%%%%%%%%%%%%%%%%%%%%%%%

\section{Dual representations}
\label{dual repr section}

In this section we provide dual representation results for convex multi-asset risk measures. In order to apply separation techniques, we assume throughout that the space of financial positions $\cX$ is a {\em locally convex} ordered topological vector space. Instead of using conjugate duality, we opt for an alternative approach which exploits the particular structure of risk measures and is based on a dual representation of the underlying acceptance set.

%%%%%%%%%%%%%%%%%%%%%%%%%%%%%%%%%%

\subsection{Dual representation of the augmented acceptance set}

By the Reduction Lemma a multi-asset risk measure of the form $\rho_{\cA,\elig,\pi}$ can be represented as a single asset risk measure with respect to the augmented acceptance set $\cA+\elig_0$. We first provide a dual representation for $\cA+\elig_0$ and derive, in a second step, a dual representation for $\rho_{\cA,\elig,\pi}$.

\medskip

We start by recalling a useful equivalent formulation of the Hahn-Banach theorem based on the notion of a support function. Here, the \textit{(lower) support function} of a set $\cA\subset\cX$ is the subadditive and positively homogenous map $\sigma_\cA:\cX'\to\R\,\cup\{-\infty\}$ defined by
\begin{equation}
\sigma_\cA(\psi):=\inf\limits_{A\in\cA}\psi(A)\,.
\end{equation}
Note that $\sigma_\cA=\sigma_{\Closure(\cA)}$ holds. The domain of finiteness of the support function is called the \textit{barrier cone} of $\cA$ and denoted by $B(\cA)$, i.e.
\begin{equation}
B(\cA):=\{\,\psi\in\cX': \, \sigma_\cA(\psi)>-\infty\,\}\,.
\end{equation}
A succinct and useful equivalent formulation of the Hahn-Banach theorem (see~\cite{AubinEkeland1984}) reads as follows: For every closed, convex subset $\cA$ of $\cX$, we have
\begin{equation}
\label{hahn-banach}
\cA = \bigcap\limits_{\psi\in\cX'} \HS(\psi,\sigma_\cA(\psi)) = \bigcap\limits_{\psi\in B(\cA)} \HS(\psi,\sigma_\cA(\psi))\,.
\end{equation}

Note that if $\cA$ is a cone we have
\begin{equation}
B(\cA)=\{\,\psi\in\cX': \, \sigma_\cA(\psi)=0\,\}=\{\,\psi\in\cX': \, \psi(A)\geq0, \;\, \forall \ A\in\cA\,\}\,.
\end{equation}
Moreover, if $\cA$ is an acceptance set, then Lemma~\ref{halfspaces containing acceptance sets} implies that
\begin{equation}
B(\cA)\subset\cX'_+\,.
\end{equation}

We now describe the support function of the augmented acceptance set $\cA+\elig_0$. By $\elig^\bot$ we denote the \textit{annihilator} of a subspace $\elig\subset\cX$, i.e. $\elig^\bot:=\{\,\psi\in\cX': \, \psi(X)=0, \;\, \forall \ X\in\elig\,\}$.

\medskip

\begin{lemma}
\label{basic remarks support function}
Let $(\cA,\elig,\pi)$ be a risk measurement regime. Then
\begin{equation}
B(\cA+\elig_0)=B(\cA)\cap\elig_0^\bot
\end{equation}
and for any $\psi\in\cX'$
\begin{equation}
\sigma_{\cA+\elig_0}(\psi)=\left\{
\begin{array}{l l}
\sigma_\cA(\psi) & \ \ \ \mbox{\rm if} \ \ \psi\in\elig_0^\bot\,,\\
-\infty & \ \ \ \mbox{\rm otherwise}\,.
\end{array}
\right.
\end{equation}
\end{lemma}

\medskip

Given a subspace $\elig$ of $\cM$, we denote by $\extens$ the set of all positive, continuous, linear extensions of the pricing functional $\pi:\elig\to\R$ to the whole space $\cX$, i.e.
\begin{equation}
\extens:=\{\,\psi\in\cX'_+: \, \psi(Z)=\pi(Z), \;\, \forall \ Z\in\elig\,\}\,.
\end{equation}
Note that, for any positive $U\in\elig$ with $\pi(U)=1$, we have
\begin{equation}
\label{equivalent extensions}
\extens=\{\,\psi\in\elig_0^\bot: \, \psi(U)=1\,\}\,.
\end{equation}
Note also that we are considering extensions of the pricing functional $\pi$ restricted to $\cS$. In particular, the elements in $\extens$ need not be extensions of $\pi:\cM\to\R$.

\medskip

\begin{remark}
Assume the market satisfies the ``No Free Lunch'' condition $\cX_+\cap\Closure(\elig_0-\cX_+)=\{0\}$. Then, using a separation argument, it is not difficult to show that $\extens$ is nonempty. For background on this condition we refer to Kreps~\cite{Kreps1981} and the discussion in Clark~\cite{Clark1993}.
\end{remark}

\medskip

\begin{theorem}
\label{key corollary}
Let $(\cA,\elig,\pi)$ be a convex risk measurement regime. The following statements hold:
\begin{enumerate}[(i)]
\item if $B(\cA)\cap\extens$ is nonempty, then
  \begin{equation}
  \label{acceptance and multiple 2}
  \Closure(\cA+\elig_0)
  =\bigcap\limits_{\psi\in B(\cA)\cap\,\extens} \HS(\psi,\sigma_\cA(\psi))\,;
  \end{equation}
\item if $B(\cA)\cap\extens$ is empty, then
  \begin{equation}
  \label{acceptance and multiple}
  \Closure(\cA+\elig_0)
  =\bigcap\limits_{\psi\in B(\cA)\cap\,\elig^\bot} \HS(\psi,\sigma_\cA(\psi))\,.
  \end{equation}
\end{enumerate}
\end{theorem}
\begin{proof}
Since the augmented acceptance set $\cA+\elig_0$ is convex, we can apply Hahn-Banach in the form~\eqref{hahn-banach} and use Lemma~\ref{basic remarks support function} to obtain
\begin{equation}
\label{key corollary: equation}
\Closure(\cA+\elig_0) = \bigcap\limits_{\psi\in B(\cA)\cap\,\elig_0^\bot} \HS(\psi,\sigma_\cA(\psi))\,.
\end{equation}

To prove~\textit{(i)} assume that $B(\cA)\cap\extens$ is nonempty. Take $X\in\cX$ such that $\varphi(X)\geq\sigma_\cA(\varphi)$ for all $\varphi\in B(\cA)\cap\extens$, and note that, to conclude the proof, it is sufficient to show that $\psi(X)\geq\sigma_\cA(\psi)$ for every $\psi\in B(\cA)\cap\elig_0^\bot$.

\smallskip

To show this, fix $\psi\in B(\cA)\cap\elig_0^\bot$. Note that, since $(\cA,\elig,\pi)$ is a risk measurement regime, there exists a positive payoff $U\in\elig$ such that $\pi(U)=1$. Assume first that $\psi(U)\neq0$. Then $\psi(U)>0$ and, up to a scaling by $\psi(U)$, the functional $\psi$ belongs to $B(\cA)\cap\extens$ so that $\psi(X)\geq\sigma_\cA(\psi)$. Otherwise, let $\psi(U)=0$ so that $\psi$ annihilates the whole $\elig$. In this case, take $\varphi\in B(\cA)\cap\extens$ and set $\varphi_n:=\varphi+n\psi$ for $n\in\N$. As a consequence of the superlinearity of $\sigma_\cA$, the functional $\varphi_n$ also belongs to $B(\cA)\cap\extens$. Hence,
\begin{equation}
\frac{1}{n}\varphi(X)+\psi(X)=\frac{1}{n}\varphi_n(X)\geq\frac{1}{n}\sigma_\cA(\varphi_n)
\geq\frac{1}{n}\sigma_\cA(\varphi)+\sigma_\cA(\psi)\,.
\end{equation}
Letting $n\to\infty$ we obtain $\psi(X)\geq\sigma_\cA(\psi)$, concluding the proof of~\textit{(i)}.

\smallskip

To prove~\textit{(ii)}, assume $B(\cA)\cap\extens$ is empty. Then we have
\begin{equation}
B(\cA)\cap\elig_0^\bot = B(\cA)\cap\elig^\bot
\end{equation}
and the claim follows immediately from~\eqref{key corollary: equation}, concluding the proof.
\end{proof}

\medskip

\begin{remark}
Note that if $\elig$ is one dimensional, then $\elig_0=\{0\}$. Hence, the above result provides a dual representation for the closure of a  convex acceptance set.
\end{remark}

%%%%%%%%%%%%%%%%%%%%%%%%%%%%%%%%%%%%%%%%%%

\subsection{Extending the pricing functional}
\label{extend pi section}

Consider a convex risk measurement regime $(\cA,\elig,\pi)$. As a consequence of Theorem~\ref{key corollary}, it is important to investigate the existence of positive, continuous linear extensions of the pricing functional $\pi$ that belong to the barrier cone $B(\cA)$. In the following theorem we provide equivalent conditions for such extensions to exist. Since the acceptance set $\cA$ is only required to be convex, this result is also of independent interest as it provides a generalization of classical extension results for positive functionals established by Namioka in Theorems~2.2 and~4.4 in~\cite{Namioka1957}, by Bauer in Theorem~2 in~\cite{Bauer1957}, and by Hustad in Theorem~2 in~\cite{Hustad1960}.

\medskip

\begin{theorem}
\label{ext of pi}
Let $(\cA,\elig,\pi)$ be a convex risk measurement regime and assume that $\cA\cap\elig\neq\emptyset$. The following statements are equivalent:
\begin{enumerate}[(a)]
	\item $B(\cA)\cap\extens$ is nonempty;
	\item $\pi$ is bounded from below on $\Closure(\cA+\elig_0)\cap\elig$;
	\item $\pi$ is bounded from below on $(\cA+\cU)\cap\elig$ for some neighborhood of zero $\cU$.
\end{enumerate}
\end{theorem}
\begin{proof}
We first prove that~\textit{(a)} and~\textit{(b)} are equivalent. Assume~\textit{(a)} holds and take $\psi\in\extens\cap B(\cA)$. Then, by Theorem~\ref{key corollary}, we have $\pi(Z)=\psi(Z)\geq\sigma_\cA(\psi)>-\infty$ for all $Z\in\Closure(\cA+\elig_0)\cap\elig$, implying that~\textit{(b)} holds.

\smallskip

Assume now that~\textit{(b)} holds but $\extens\cap B(\cA)$ is empty. Then $\elig\subset\Closure(\cA+\elig_0)$. Indeed, taking $W\in\cA\cap\elig\subset\Closure(\cA+\elig_0)$ we have $\sigma_\cA(\psi)\leq\psi(W)=0$ for all $\psi\in B(\cA)\cap\elig^\bot$. Consequently, $\sigma_\cA(\psi) \le 0=\psi(Z)$ for all $Z\in\elig$ and $\psi\in B(\cA)\cap\elig^\bot$, yielding $\elig\subset\Closure(\cA+\elig_0)$ by Theorem~\ref{key corollary}. However, this implies $\Closure(\cA+\elig_0)\cap\elig=\elig$, contradicting~\textit{(b)}. It follows that~\textit{(a)} and~\textit{(b)} are equivalent.

\smallskip

To see that \textit{(a)} implies \textit{(c)}, take $\psi\in\extens\cap B(\cA)$. If $\cU:=\{\,X\in\cX: \, \psi(X)>-1\,\}$, then $\pi(Z)=\psi(Z)>\sigma_\cA(\psi)-1>-\infty$ for all $Z\in(\cA+\cU)\cap\elig$.

\smallskip

We conclude by proving that \textit{(c)} implies \textit{(a)}. Take a positive $U\in\elig$ with $\pi(U)=1$. Then, $\elig=\R U+\elig_0$. Without loss of generality we can take $\cU$ to be open and convex so that $\cA+\cU$ is an open and convex acceptance set. Note that $(\cA+\cU)\cap\{\,Z\in\elig: \, \pi(Z)\leq m\,\}$ is empty for some $m<0$ and that $\{\,Z\in\elig: \, \pi(Z)\leq m\,\}=\{\,\lambda U+Z_0: \, \lambda\le m, \;\, Z_0\in\elig_0\}$. Hence, we find by separation and Lemma~\ref{halfspaces containing acceptance sets} a nonzero positive $\psi\in\cX'$ such that
\begin{equation}
\label{extension aux}
\lambda\psi(U)+\psi(Z_0)\leq\psi(A+X)
\end{equation}
for all $A\in\cA$, $X\in\cU$, $\lambda\le m$ and $Z_0\in\elig_0$. Since $\elig_0$ is a subspace,~\eqref{extension aux} implies that $\psi\in\elig_0^\bot$. Furthermore $\psi(U)>0$. Indeed, since $\psi(U)\ge 0$, we would otherwise have $\psi(U)=0$ and, hence, $\psi\in\elig^\bot$. But then, taking $A\in\cA\cap\elig$ we would obtain from~\eqref{extension aux}~that $0\leq\psi(X)$ for $X\in\cU$ which is impossible since $\cU$ is a neighborhood of zero and $\psi$ is nonzero. Rescaling $\psi$ to satisfy $\psi(U)=\pi(U)$ we have $\psi\in\extens$. Finally,~\eqref{extension aux} also implies that $\inf_{A\in\cA}\psi(A)>-\infty$ so that $\psi\in B(\cA)$.
\end{proof}

\medskip

\begin{remark}
\label{remark on acceptable elig payoffs}
Let $(\cA,\elig,\pi)$ be a convex risk measurement regime.
\begin{enumerate}[(i)]
	\item It is easy to see that, if $\cA+\elig_0$ is closed, the conditions in the previous theorem are equivalent to the ``no acceptability arbitrage'' condition stated in Lemma~\ref{lemma no acc arbitrage}, i.e. $\cA\cap\{\,Z\in\elig: \, \pi(Z)\leq m\,\}=\emptyset$ for some $m\in\R$. In particular, they are equivalent to $\rho_{\cA,\elig,\pi}(0)>-\infty$. A general criterion for $\cA+\elig_0$ to be closed is provided in Proposition~\ref{closedeness of a sum} above.
	\item Requiring that $\cA\cap\elig\neq\emptyset$ is equivalent to requiring $\rho_{\cA,\elig,\pi}(0)<\infty$. This is reasonable since the zero position should either be acceptable in the first place or capable of being made acceptable by some eligible strategy. Moreover, if $\cA\cap\elig=\emptyset$ then $\rho_{\cA,\elig,\pi}(Z)=\infty$ for all $Z\in\elig$.
\end{enumerate}
\end{remark}

%%%%%%%%%%%%%%%%%%%%%%%%%%%%%%%%%%%%%%%%%%%%%%%%%%%%%%%%%%%%%%

\subsection{Dual representation of convex multi-asset risk measures}

In this section we derive dual representation theorems for convex, lower semicontinuous multi-asset risk measures. We refer to Section~\ref{finiteness continuity section} for conditions ensuring lower semicontinuity or continuity under various assumptions on the risk measurement regime $(\cA,\elig,\pi)$.

\medskip

We start by showing that the condition $\extens\cap B(\cA)\neq\emptyset$ characterizes the nondegeneracy of convex multi-asset risk measures.

\medskip

\begin{proposition}
\label{sound theory}
Let $(\cA,\elig,\pi)$ be a convex risk measurement regime. If $\extens\cap B(\cA)=\emptyset$, then $\rho_{\cA,\elig,\pi}$ cannot take finite values at any point $X\in\cX$ of lower semicontinuity.
%More precisely, $\rho_{\cA,\elig,\pi}(X)=-\infty$ if $X\in\overline{\cA+\elig_0}$ and $\rho_{\cA,\elig,\pi}(X)=\infty$ if $X\not\in\overline{\cA+\elig_0}$.
\end{proposition}
\begin{proof}
Let $X$ be a point of lower semicontinuity for $\rho_{\cA,\elig,\pi}$. By Proposition~\ref{pointwise semicontinuity} we can assume that $\cA+\elig_0$ is closed. Take a positive $U\in\elig$ with $\pi(U)=1$. As a consequence of Theorem~\ref{key corollary}, we have for any $m\in\R$ that $X+mU\in\cA+\elig_0$ if and only if $X\in\cA+\elig_0$. Hence, the Reduction Lemma implies that $\rho_{\cA,\elig,\pi}(X)=-\infty$ if $X\in\cA+\elig_0$ and $\rho_{\cA,\elig,\pi}(X)=\infty$ if $X\notin\cA+\elig_0$.
\end{proof}

\medskip

As a corollary of the above result we obtain a characterization of when a convex, (globally) lower semicontinuous multi-asset risk measure never takes the value $-\infty$. In particular, this is the case whenever there is no ``acceptability arbitrage'' as defined in Section~\ref{finiteness continuity section}.

\medskip

\begin{corollary}
\label{corollary on pointwise representation}
Let $(\cA,\elig,\pi)$ be a convex risk measurement regime with $\cA\cap\elig\neq\emptyset$. Assume $\rho_{\cA,\elig,\pi}$ is lower semicontinuous. Then the following statements are equivalent:
\begin{enumerate}[(a)]
  \item $\rho_{\cA,\elig,\pi}(X)\in\R$ for some $X\in\cX$;
  \item $\extens\cap B(\cA)$ is nonempty;
  \item $\rho_{\cA,\elig,\pi}(X)>-\infty$ for every $X\in\cX$;
  \item $\rho_{\cA,\elig,\pi}(0)>-\infty$.
\end{enumerate}
\end{corollary}
\begin{proof}
By Proposition~\ref{pointwise semicontinuity} we can assume that $\cA+\elig_0$ is closed. Clearly \textit{(a)} implies \textit{(b)} by Proposition~\ref{sound theory}. Assume \textit{(b)} holds and take $X\in\cX$ and $\psi\in\extens\cap B(\cA)$. Take a positive $U\in\elig$ with $\pi(U)=1$. Since $\psi(X)+m=\psi(X+mU)\geq\sigma_\cA(\psi)$ cannot hold for every $m\in\R$, Theorem~\ref{key corollary} implies that $X+mU\notin\cA+\elig_0$ for some $m\in\R$. Hence, $\rho_{\cA,\elig,\pi}(X)>-\infty$ by the Reduction Lemma, proving \textit{(c)}. Clearly, \textit{(c)} implies \textit{(d)} which implies \textit{(a)} because we have $\rho_{\cA,\elig,\pi}(0)<\infty$ as a consequence of the assumption $\cA\cap\elig\neq\emptyset$.
\end{proof}

\medskip

\begin{remark}
This corollary can be seen as a sharper version, for the case of risk measures, of a well-known result stating that, on a locally convex topological vector space, lower semicontinuous convex functions which are not identical to $\infty$ take some finite value if and only if they never assume the value $-\infty$; see for instance Proposition~2.4 in~\cite{EkelandTemam1976}. For a risk measure with risk measurement regime $(\cA,\elig,\pi)$ satisfying $\cA\cap\elig\neq\emptyset$, this is the case if and only if it does not assume the value $-\infty$ at $0$.
\end{remark}

\medskip

We are now ready to prove the following version of the dual representation for convex multi-asset risk measures. As already mentioned, the proof draws on the dual representation of the augmented acceptance set obtained in Theorem~\ref{key corollary}.

\medskip

\begin{theorem}
\label{pointwise repr formula for multiple assets}
Let $(\cA,\elig,\pi)$ be a convex risk measurement regime and assume $\rho_{\cA,\elig,\pi}$ is lower semicontinuous at $X\in\cX$. If $\extens\cap B(\cA)$ is nonempty, then
\begin{equation}
\label{representation formula eqn 2 multiple}
\rho_{\cA,\elig,\pi}(X)=\sup_{\psi\in\extens}\left\{\sigma_\cA(\psi)-\psi(X)\right\}\,.
\end{equation}
In particular, if $\elig$ is generated by a nonzero $U\in\cX_+$ and $\psi(U)>0$ for some $\psi\in B(\cA)$, then
\begin{equation}
\label{representation formula eqn 2 multiple, single}
\rho_{\cA,U,\pi}(X)=\sup_{\psi\in\cX'_+, \,\psi(U)=\pi(U)}\left\{\sigma_\cA(\psi)-\psi(X)\right\}\,.
\end{equation}
\end{theorem}
\begin{proof}
Note that the augmented acceptance set $\cA+\elig_0$ is convex and that, by Proposition~\ref{pointwise semicontinuity}, we may assume that it is closed. Fix now a positive $U\in\elig$ with $\pi(U)=1$. Then the Reduction Lemma and Theorem~\ref{key corollary} imply
\begin{eqnarray}
\rho_{\cA,\elig,\pi}(X)
&=&
\inf\{\,m\in\R: \, X+mU\in\cA+\elig_0\,\} \\
&=&
\inf\Big\{\,m\in\R: \, \psi(X)+m\,\psi(U)\geq\sigma_{\cA}(\psi), \;\, \forall \ \psi\in\extens\,\Big\}\,.
\end{eqnarray}
Since $\psi(U)=\pi(U)=1$ for all $\psi\in\extens$, the representation~\eqref{representation formula eqn 2 multiple} immediately follows. Finally, if $\elig$ is spanned by a nonzero, positive $U\in\cX$, then $\psi\in\extens$ is equivalent to $\psi\in\cX'_+$ and $\psi(U)=\pi(U)$, and the corresponding dual representation follows from~\eqref{representation formula eqn 2 multiple}.
\end{proof}

\medskip

In case of coherent risk measurement regimes we obtain a simplified representation as a consequence of the properties of the corresponding support function.

\medskip

\begin{corollary}
\label{representation of coherent}
Assume $(\cA,\elig,\pi)$ is a coherent risk measurement regime and $\rho_{\cA,\elig,\pi}$ is lower semicontinuous at $X\in\cX$. If $\extens\cap B(\cA)$ is nonempty, then
\begin{equation}
\rho_{\cA,\elig,\pi}(X)=\sup\{\,\psi(-X): \, \psi\in\extens, \;\, \psi(A)\geq0, \ \forall \ A\in\cA\}\,.
\end{equation}
In particular, if $\elig$ is generated by a nonzero $U\in\cX_+$ and $\psi(U)>0$ for some $\psi\in B(\cA)$, then
\begin{equation}
\rho_{\cA,U,\pi}(X)=\sup\{\,\psi(-X): \, \psi\in\cX'_+, \;\, \psi(U)=\pi(U) \textrm{ and } \psi(A)\geq0, \ \forall \ A\in\cA\,\}\,.
\end{equation}
\end{corollary}

\medskip

\begin{remark}
\label{separation of roles in repr formula}
The above dual representation~\eqref{representation formula eqn 2 multiple} highlights the different roles played by $\cA$, $\elig$ and $\pi$ in the determination of the risk measure $\rho_{\cA,\elig,\pi}$. The acceptance set $\cA$ determines, through the support function $\sigma_\cA$, the \textit{objective function} and the space $\elig$ together with the pricing functional $\pi$ determines the \textit{optimization domain} $\extens$. In particular, modifying the eligible space while maintaining the acceptance set, only requires changing the optimization domain, but not the objective function, simplifying the practical implementation when alternative choices for the eligible space need to be considered.
\end{remark}

\medskip

\begin{remark}
In case of finite-valued risk measures, the representation~\eqref{representation formula eqn 2 multiple} is equivalent to that obtained by Frittelli and Scandolo in~\cite{FrittelliScandolo2006} by means of conjugate duality methods. In particular, it is not difficult to show that
\begin{equation}
\sigma_\cA(\psi)=-\rho_{\cA,\elig,\pi}^\ast(-\psi)=
-\sup_{X\in\cX}\{-\psi(X)-\rho_{\cA,\elig,\pi}(X)\} \ \ \ \mbox{for all} \ \psi\in\extens\,,
\end{equation}
where $\rho_{\cA,\elig,\pi}^\ast$ denotes the standard conjugate function of $\rho_{\cA,\elig,\pi}$.

\smallskip

Note that the condition $B(\cA)\cap\extens\neq\emptyset$, which is key to ensure nontrivial representations, appears naturally when investigating the structure of acceptance sets, which was the starting point of our analysis of dual representations. In this sense, our approach is more geometrical in character and seems to lead in a very natural way to consider the interplay between the acceptance set (via its barrier cone) and the eligible space (via the set of positive extensions). Additional evidence of the advantages of a more geometrical approach is provided in Section~\ref{shortfall section}, see in particular Remark~\ref{advantage geometrical approach with arai}.
\end{remark}

\medskip

\begin{remark}
\label{remark: attainability}
A natural and important question to ask is when the supremum in the representation formula~\eqref{representation formula eqn 2 multiple} is attained at a point $X\in\cX$. This is always the case if $\rho_{\cA,\elig,\pi}$ is finite and continuous at $X$, as can be derived from Theorem~7.12 in~\cite{AliprantisBorder2006}. This can also be shown directly  exploiting the particular structure of $\rho_{\cA,\elig,\pi}$. Indeed, let $X$ be a point of finiteness and continuity for $\rho_{\cA,\elig,\pi}$, and take a positive payoff $U\in\elig$ with $\pi(U)=1$. Then Proposition~\ref{pointwise semicontinuity} implies that $\Interior(\cA+\elig_0)\neq\emptyset$ and $X+mU\in\Interior(\cA+\elig_0)$ for every $m>\rho_{\cA,\elig,\pi}(X)$. Fix such an $m$. Since $X+\rho_{\cA,\elig,\pi}(X)U$ belongs to the boundary of $\cA+\elig_0$, it follows from Lemma~7.7 in~\cite{AliprantisBorder2006} that it is also a support point of $\cA+\elig_0$. Let $\psi\in\cX'$ be the corresponding supporting functional which, by Lemma~\ref{halfspaces containing acceptance sets}, must be positive. Since $\psi(X+\rho_{\cA,\elig,\pi}(X)U)=\sigma_{\cA+\elig_0}(\psi)$, we must have $\psi\in\elig_0^\bot$. Moreover $\psi(U)>0$, as otherwise $\psi(X+mU)=\psi(X+\rho_{\cA,\elig,\pi}(X)U)=\sigma_{\cA+\elig_0}(\psi)$ contradicting $X+mU\in\Interior(\cA+\elig_0)$. Hence we can assume that $\psi(U)=1$, concluding that $\rho_{\cA,\elig,\pi}(X)=\sigma_\cA(\psi)-\psi(X)$. This shows that the supremum in~\eqref{representation formula eqn 2 multiple} is attained at $X$.
\end{remark}

\section{Applications}
\label{fields of application}

In this final section we apply our previous results to the conical market model adopted by Hamel, Heyde, and Rudloff in~\cite{HamelHeydeRudloff2011} as the underlying setting for their set-valued risk measures, and to the market model considered by Arai in~\cite{Arai2011} in connection to superhedging problems. Moreover, we show that multi-asset risk
measures appear naturally in the context of optimal risk sharing across different business lines.

%%%%%%%%%%%%%%%%%%%%%%%%%%%%%%%%%%%%%%%%%%%%

\subsection{Risk measures on conical market models}
\label{set valued section}

In this section we investigate the link between multi-asset risk measures and set-valued risk measures as introduced by Hamel, Heyde and Rudloff in~\cite{HamelHeydeRudloff2011}. We start by briefly recalling the setting of that paper.

\medskip

We consider a one-period economy with dates $t=0$ and $t=T$ where uncertainty is captured by a probability space $(\Omega,\cF,\probp)$. We will consider risk measures for random portfolios of $d$ traded assets. Random portfolios are described by random vectors $X=(X_1,\dots,X_d)^t\in L^p_d$ for some $1\leq p\leq \infty$, where $X_i\in L^p$ represents the (random) number of units of asset $i$ held at time $T$. Note that the space $L^p_d$ inherits the order structure of $\R^d$ in the almost surely sense, see also Remark~\ref{remark: order structure standard spaces}. The topology on $L^p_d$ is the usual norm topology if $p<\infty$ and the $\sigma(L^\infty_d,L^1_d)$ topology if $p=\infty$.

\medskip

Let $A\subset L^p_d$ be an acceptance set. A random portfolio can be made acceptable by adding ``eligible'' deterministic portfolios at time $0$. The space of eligible portfolios is represented by a subspace $M\subset\R^d$, which is identified with the subspace of $L^p_d$ given by $\{\,(u_1 1_\Omega,\dots,u_d 1_\Omega)^t: \, u\in M\,\}$. Following~\cite{HamelHeydeRudloff2011}, we assume that $M$ contains some nonzero positive element so that $(A,M)$ is a risk measurement regime. The associated {\em set-valued risk measure} is the set-valued map $R_A$ defined by
\begin{equation}
R_A(X):=\{\,u\in M: \, X+u\in A\,\} \ \ \ \mbox{for} \ X\in L^p_d\,.
\end{equation}
For set-valued risk measures to be useful as capital requirements it is necessary to specify a procedure through which, roughly speaking, an optimal element in $R_A(X)$ can be chosen. Such a procedure is called {\em scalarization} and is explained in Section~5 of~\cite{HamelHeydeRudloff2011} and, in greater detail, in~\cite{HamelHeyde2010}. To illustrate this procedure, we recall the notion of the solvency cone.

\medskip

Let $(\pi_{ij}^0)$ be a $d\times d$ \textit{bid-ask matrix} as defined by Schachermayer~\cite{Schachermayer2004}, i.e. $\pi_{ij}^0$ represents the number of units of asset $i$ that are required to purchase one unit of asset $j$ at time $0$. The \textit{solvency cone} $K_0$ at time $0$ is the convex cone in $\R^d$ consisting of all \textit{solvent} deterministic portfolios, i.e. all portfolios $u\in\R^d$ admitting a $d\times d$ matrix $(\alpha_{ij})$ with $\alpha_{ij}\geq0$ and $\alpha_{ii}=0$ for $i,j=1,\dots,d$ such that
\begin{equation}
\label{solvency cone eq}
u_i+\sum^{d}_{j=1}\alpha_{ij}\pi_{ij}^0-\sum^{d}_{j=1}\alpha_{ji}\geq0 \ \ \ \mbox{for all} \ i=1,\dots,d\,,
\end{equation}
where the coefficient $\alpha_{ij}$ is to be interpreted as the number of units of asset $j$ that are used to modify the position $i$ in the portfolio. Hence, the solvency cone $K_0$ contains all portfolios which can be converted at time $0$ into portfolios with nonnegative components.

\medskip

The elements in the dual cone
\begin{equation}
K_0^+:=\Big\{\,\xi\in\R^d: \, \sum^{d}_{i=1}\xi_iu_i\geq0, \;\, \forall \ u\in K_0\,\Big\}
\end{equation}
are called \textit{consistent pricing systems}, and are easily seen to have strictly positive components when nonzero. The scalarization of a set-valued risk measure $R_A$ at $X\in L^p_d$ by means of $\xi\in K_0^+$ is then defined as
\begin{equation}
\varphi_{R_A,\xi}(X):=\inf\Big\{\,\sum^{d}_{i=1}\xi_iu_i: \, u\in R_A(X)\,\Big\}\,.
\end{equation}

We start by showing the intimate link between scalarized set-valued risk measures and multi-asset risk measures. This allows for a fruitful exchange of results between the theory developed in~\cite{HamelHeydeRudloff2011} and our multi-asset framework. In particular, we show that scalarized set-valued risk measures appear naturally as multi-asset risk measures. Hence, the results in this paper can be used to prove finiteness and continuity properties and to provide dual representations for scalarized set-valued risk measures. Note that, in this respect, only basic results are provided in~\cite{HamelHeydeRudloff2011}.

\medskip

\begin{proposition}
\label{from set-valued to multi}
For every $\xi\in K_0^+$ there exists a linear functional $\pi:M\to\R$ such that
\begin{equation}
\label{equation: from set valued to multi}
\varphi_{R_A,\xi}(X)=\rho_{A,M,\pi}(X) \ \ \ \mbox{for every} \ X\in L^p_d\,.
\end{equation}
\end{proposition}
\begin{proof}
Define $\pi:M\to\R$ by $\pi(u):=\sum^{d}_{i=1}\xi_iu_i$. Then it is immediate to see that
\begin{equation}
\varphi_{R_A,\xi}(X)=\inf\{\,\pi(u): \, u\in M, \;\, X+u\in A\,\}=\rho_{A,M,\pi}(X)
\end{equation}
holds for every $X\in L^p_d$.
\end{proof}

\medskip

A converse of the previous result is also possible if we consider multi-asset risk measures on $L^p$ with respect to finite-dimensional eligible spaces.

\medskip

\begin{proposition}
Let $(\cA,\elig,\pi)$ be a risk measurement regime in $L^p$, $1\leq p\leq\infty$, and assume $\dim(\elig)=d$. Then there exist an acceptance set $A\subset L^p_{d+1}$, a linear space $M\subset\R^{d+1}$, a convex cone $K_0\subset\R^{d+1}$ and $\xi\in K_0^+$ such that
\begin{equation}
\label{from multi to set valued}
\rho_{\cA,\elig,\pi}(X)=\varphi_{R_A,\xi}\big((0,\dots,0,X)^t\big) \ \ \ \mbox{for every} \ X\in L^p\,.
\end{equation}
\end{proposition}
\begin{proof}
Let $\elig$ be the span of $Z_1,\dots,Z_d\in L^p$, and define the acceptance set
\begin{equation}
A:=\Big\{\,X\in L^p_{d+1}: \, \sum^{d}_{i=1}X_iZ_i+X_{d+1}\in\cA\,\Big\}\,.
\end{equation}
Moreover, set
\begin{equation}
M:=\{\,u\in\R^{d+1}: \, u_{d+1}=0\,\} \ \ \ \mbox{and} \ \ \ K_0:=\Big\{\,u\in\R^{d+1}: \, \sum^{d}_{i=1}\pi(Z_i)u_i+u_{d+1}\geq0\,\Big\}\,.
\end{equation}
Taking $\xi:=(\pi(Z_1),\dots,\pi(Z_d),1)^t\in K_0^+$, it is easy to see that~\eqref{from multi to set valued} holds for every $X\in L^p$.
\end{proof}

\medskip

\begin{remark}
The transition from a scalarized set-valued risk measure to a multi-asset risk measure in~\eqref{equation: from set valued to multi} is fairly natural in that the original underlying framework can be retained: The underlying space, the acceptance set, and the eligible space do not change. By contrast, the transition from a multi-asset risk measure to a scalarized set-valued risk measure seems to be more ``formal'' in character since we need to artificially enlarge the dimension of the eligible space by a ``cash'' asset and define a new acceptance set. Moreover, the new cash asset essentially plays the role of a num\'{e}raire asset since all positions $X\in L^p$ are expressed as random vectors $(0,\dots, 0,X)^t$.
\end{remark}

\medskip

Next, we provide a dual representation for scalarized set-valued risk measures based on the multi-asset dual representation~\eqref{representation formula eqn 2 multiple}. As usual, we identify the dual of $L^p_d$ with the space $L^q_d$ where $1\leq q\leq\infty$ satisfies $1/p+1/q=1$. Moreover, we denote by $\mathcal{P}^q_d$ the set of all $d$-dimensional vectors $\probq$ whose components $\probq_i$ are probability measures on $(\Omega,\cF)$ that are absolutely continuous with respect to $\probp$ and such that $\frac{d\probq_i}{d\probp}\in L^q$. In order to highlight the link with the dual representations for set-valued risk measures obtained in~\cite{HamelHeydeRudloff2011}, we adopt the same set of dual variables.

\medskip

\begin{proposition}
\label{dual repr for scal set valued}
Assume $A\subset L^p_d$ is convex. Let $\xi\in K_0^+$ and define $\pi(u):=\sum^{d}_{i=1}\xi_i u_i$ on $M$. If $B(A)\cap\mathscr{E}_\pi(M)$ is nonempty and $\varphi_{R_A,\xi}$ is lower semicontinuous at $X$, then
%\begin{equation}
%\varphi_{R_A,\pi}(X)=\sup_{Y\in\mathscr{D}^q}\left\{\sum^{d}_{i=1}\E[-Y_iX_i]+\inf_{Z\in %A}\left(\sum^{d}_{i=1}\E[Y_iZ_i]\right)\right\}
%\end{equation}
%where
%\begin{equation}
%\mathscr{D}^q:=\left\{Y\in(L^q_d)_+ \,; \ \sum^{d}_{i=1}u_i\E[Y_i]=\sum^{d}_{i=1}\pi_i u_i, \ \forall \ u\in M\right\}\,.
%\end{equation}
%Equivalently, we have
\begin{equation}
\label{dual repr set valued}
\varphi_{R_A,\xi}(X)=\sup_{(\probq,w)\in\mathcal{D}^q}\Big\{\sum^{d}_{i=1}w_i\E_{\probq_i}[-X_i]+\inf_{Y\in A}\sum^{d}_{i=1}w_i\E_{\probq_i}[Y_i]\Big\}
\end{equation}
where
\begin{equation}
\mathcal{D}^q:=\Big\{\,(\probq,w)\in\mathcal{P}^q_d\times\R^d_+: \, \sum^{d}_{i=1}w_iu_i=\sum^{d}_{i=1}\xi_i u_i, \;\, \forall \ u\in M\,\Big\}\,.
\end{equation}
\end{proposition}
\begin{proof}
By~\eqref{representation formula eqn 2 multiple} and~\eqref{equation: from set valued to multi} we immediately obtain
\begin{equation}
\label{dual repr set valued, auxiliary}
\varphi_{R_A,\xi}(X)=\sup_{\psi\in\mathscr{E}_\pi(M)}\{\sigma_A(\psi)-\psi(X)\}\,.
\end{equation}
Note that every functional $\psi$ in $\mathscr{E}_\pi(M)$ can be identified with a positive element $W\in L^q_d$ such that $\psi(X')=\sum^{d}_{i=1}\E[W_iX'_i]$ for every $X'\in L^p_d$ and $\sum^{d}_{i=1}u_i\E[W_i]=\sum^{d}_{i=1}\xi_iu_i$ for all $u\in M$. In turn, every $W$ of this form can be identified with $(\probq,w)\in\mathcal{D}^q$ by setting $w_i:=\E[W_i]$ and $\frac{d\probq_i}{d\probp}:=\frac{1}{w_i}W_i$, or $\probq_i:=\probp$ if $w_i=0$, for any $i=1,\dots,d$. This concludes the proof.
\end{proof}

\medskip

The random character of the market at time $t=T$ is reflected by the fact that the bid-ask matrix $(\pi_{ij}^T)$ at time $T$ is random. When defining the corresponding ``random'' solvency cone $K_T$ at time $T$ we may proceed as in~\eqref{solvency cone eq} but requiring that the ``transition'' matrix $(\alpha_{ij})$ is also random. The convex cone $K_T^+$ is defined analogously. Using the notation of~\cite{HamelHeydeRudloff2011}, we denote by $L^p_d(K_T)$, respectively $L^p_d(K_T^+)$, the convex cone of all $X\in L^p_d$ such that $X\in K_T$, respectively $X\in K_T^+$, almost surely.

\medskip

Consider an acceptance set $A\subset L^p_d$, and let $X\in L^p_d$ be a random portfolio of assets. From a capital adequacy perspective, if we want to account for the possibility of trading at time $t=T$ it is natural to consider the set
\begin{equation}
\label{KT compatibility}
R_{A+L^p_d(K_T)}(X)=\{\,u\in M: \, X+u\in A+L^p_d(K_T)\,\}\,.
\end{equation}
Indeed, the condition $X+u\in A+L^p_d(K_T)$ means that $X+u$ will be exchangeable at time $t=T$ into an acceptable portfolio, after paying the transaction costs defined by $K_T$. For this reason, the authors in~\cite{HamelHeydeRudloff2011} have considered acceptance sets $A$ which are \textit{$K_T$-compatible}, i.e. such that $A=A+L^p_d(K_T)$. We conclude this section providing a dual representation for scalarized set-valued risk measures satisfying this compatibility condition.

\medskip

\begin{corollary}
\label{dual repr scal compatible set valued}
Assume $A\subset L^p_d$ is convex and $K_T$-compatible. For $\xi\in K_0^+$ define $\pi(u):=\sum^{d}_{i=1}\xi_i u_i$ on $M$. If $B(A)\cap\mathscr{E}_\pi(M)$ is nonempty and $\varphi_{R_A,\pi}$ is lower semicontinuous at $X$, then
%\begin{equation}
%\varphi_{R_A,\pi}(X)=\sup_{Y\in\mathscr{D}^q}\left\{\sum^{d}_{i=1}\E[-Y_iX_i]+\inf_{Z\in A}\left(\sum^{d}_{i=1}\E[Y_iZ_i]\right)\right\}
%\end{equation}
%where
%\begin{equation}
%\mathscr{D}^q:=\left\{Y\in(L^q_d)_+ \,; \ \sum^{d}_{i=1}\E[Y_iV_i]\geq0, \ \forall \ V\in L^p_d(K_T), \ \sum^{d}_{i=1}u_i\E[Y_i]=\sum^{d}_{i=1}\pi_i u_i, \ \forall \ u\in M\right\}\,.
%\end{equation}
%Equivalently, we have
\begin{equation}
\varphi_{R_A,\xi}(X)=\sup_{(\probq,w)\in\mathcal{D}^q(K_T)}\Big\{\sum^{d}_{i=1}w_i\E_{\probq_i}[-X_i]+\inf_{Y\in A}\sum^{d}_{i=1}w_i\E_{\probq_i}[Y_i]\Big\}
\end{equation}
where
\begin{equation}
\mathcal{D}^q(K_T):=\Big\{\,(\probq,w)\in\mathcal{D}^q: \, \left(w_1\frac{d\probq_1}{d\probp},\dots,w_d\frac{d\probq_d}{d\probp}\right)^t\in L^p_d(K^+_T)\,\Big\}\,.
\end{equation}
\end{corollary}
\begin{proof}
Since $A$ is $K_T$-compatible, we have $\sigma_A=\sigma_A+\sigma_{L^p_d(K_T)}$. Hence, we can restrict the optimization domain in~\eqref{dual repr set valued, auxiliary} to all functionals $\psi\in B(L^p_d(K_T))\cap\mathscr{E}_\pi(M)$. But $L^p_d(K_T)$ is a cone and therefore $\psi\in B(L^p_d(K_T))$ if and only if $\sigma_{L^p_d(K_T)}(\psi)=0$, or equivalently $\psi(Z)\geq0$ for every $Z\in L^p_d(K_T)$. Using the dual variables in~\eqref{dual repr set valued}, we see that this is equivalent to replacing $\mathcal{D}^q$ with $\mathcal{D}^q(K_T)$.
\end{proof}

%%%%%%%%%%%%%%%%%%%%%%%%%%%%%%%%%%%%%%%%%%%%

\subsection{Shortfall risk measures and superhedging price}
\label{shortfall section}

In this section we focus on the superhedging problem studied by Arai in~\cite{Arai2011}. Based on our approach to dual representations, we provide a sharper dual representation of the superhedging price defined in that paper. For the background on Orlicz hearts and Orlicz spaces we refer to~\cite{EdgarSucheston1992}.

\medskip

Fix a filtered probability space $(\Omega,\cF,(\cF_t),\probp)$, for $0\leq t\leq T$, and let $S=(S_t)$ be a $d$-dimensional semimartingale representing the price dynamics of $d$ given assets. We assume $(\cF_t)$ satisfies the usual conditions and $\cF_T=\cF$. Denote by $\Theta$ a convex set of $d$-dimensional, predictable, $S$-integrable processes $\vartheta=(\vartheta_t)$. The elements of $\Theta$ will be called \textit{admissible strategies}. In addition, consider a \textit{loss} function $\ell:\R\to\R$, which is assumed to be nonconstant, increasing, convex, and such that $\ell(0)=0$. The reference space is taken to be the Orlicz heart $H^{\Phi}$ associated to the Orlicz function $\Phi$ defined by $\Phi(x):=\ell(\left|x\right|)$ for $x\in\R$. From now on, we assume that $\int^{T}_{0}\vartheta_t dS_t\in H^\Phi$ for every admissible strategy $\vartheta\in\Theta$.

\medskip

The main goal in Arai~\cite{Arai2011} is to provide dual representations, under specific assumptions on the class $\Theta$, for the map $\rho_\ell:H^{\Phi}\to\overline{\R}$ defined by
\begin{equation}
\rho_\ell(X):=\inf\bigg\{\,m\in\R: \, \exists \ \vartheta\in\Theta, \;\, \E\Big[\ell\big(-X-m-\int^{T}_{0}\vartheta_t dS_t\big)\Big]\leq\alpha\,\bigg\}
\end{equation}
where $\alpha>0$ is a pre-specified loss level. Note that the map $\rho_\ell$ is a standard cash-additive risk measure. Indeed, if we introduce the convex acceptance set
\begin{equation}
\cA_\ell:=\{\,X\in H^{\Phi}: \, \E[\ell(-X)]\leq\alpha\,\}
\end{equation}
and the convex set
\begin{equation}
\label{set terminal integrals}
\cC:=\Big\{\,\int^{T}_{0}\vartheta_t dS_t: \, \vartheta\in\Theta\,\Big\}\,,
\end{equation}
it is easy to see that $\rho_\ell=\rho_{\cA_\ell-\cC,U,\pi}$ with $U:=1_\Omega$ and $\pi(U):=1$.

\medskip

The financial motivation for studying the map $\rho_\ell$ is given by the fact that the quantities $-\rho_\ell(X)$ and $\rho_\ell(-X)$ can be interpreted as pricing bounds for the claim $X\in H^{\Phi}$ which are compatible with the absence of ``good deals''. We refer to~\cite{Arai2011} for a detailed explanation.

\medskip

The main dual representation provided under the standing assumption $\rho_\ell(0)>-\infty$ is Proposition 3.5 in~\cite{Arai2011}. This representation is then specified to various situations including the case where $\Theta$ is a linear space (Section 4 in~\cite{Arai2011}) and $\Theta$ is the convex cone of $W$-admissible strategies (Section 5 in~\cite{Arai2011}). The key ingredient for the dual representation in the $W$-admissible case is Lemma 5.1 in that paper.

\medskip

Our objective is to show that this key lemma holds for \textit{any} choice of the class of admissible strategies $\Theta$. As a result, we provide a general dual representation for $\rho_\ell$ sharpening Proposition 3.5 in~\cite{Arai2011}.

\medskip

Given a map $f:\R\to\R$, we denote by $f^\ast:\R\to\R\cup\{\infty\}$ the Fenchel conjugate of $f$ defined by
\begin{equation}
f^\ast(y):=\sup_{x\in\R}\{xy-f(x)\}\,.
\end{equation}
Recall that the dual space $(H^\Phi)'$ can be identified with the Orlicz space $L^{\Phi^\ast}$. Moreover, we denote by $\cP^\Phi$ the set of all probability measures $\probq$ on $(\Omega,\cF)$ which are absolutely continuous with respect to $\probp$ and such that $\frac{d\probq}{d\probp}\in L^{\Phi^\ast}$.

\medskip

\begin{proposition}
\label{general dual repr arai}
If $\rho_\ell(0)>-\infty$, then $\rho_\ell$ is finitely valued and continuous on $H^\Phi$. Moreover, for every $X\in H^\Phi$ we have
\begin{equation}
\label{dual repr loss risk measure}
\rho_\ell(X)=\max_{\probq\in\cP^\Phi}\bigg\{\E_\probq[-X]-\sup_{\vartheta\in\Theta}\E_\probq\Big[\int^{T}_{0}\vartheta_t dS_t\Big]-\inf_{\lambda>0}\frac{1}{\lambda}\Big\{\alpha+\E\Big[\ell^\ast\Big(\lambda\frac{d\probq}{d\probp}\Big)\Big]\Big\}\bigg\}\,.
\end{equation}
\end{proposition}
\begin{proof}
Since $\alpha>0$, it follows from Lemma 4.5 in~\cite{FarkasKochMunari2012a} that $\cA_\ell$ has nonempty interior, hence the convex acceptance set $\cA_\ell-\cC$ also has nonempty interior. Moreover, note that $1_\Omega$ is a strictly positive element in $H^\Phi$. Since $\rho_\ell(0)>-\infty$, we can apply Proposition~\ref{convex strictly positive} to ensure that $\rho_\ell$ is finitely valued and continuous. Furthermore, by the dual representation in Theorem~\ref{pointwise repr formula for multiple assets} we have for all $X\in H^\Phi$
\begin{eqnarray*}
\rho_\ell(X)
&=&
\sup_{\psi\in(H^\Phi)'_+, \psi(1_\Omega)=1}\{\sigma_{\cA_\ell}(\psi)+\sigma_{-\cC}(\psi)-\psi(X)\} \\
&=&
\sup_{\probq\in\cP^\Phi}\bigg\{-\inf_{\lambda>0}\frac{1}{\lambda}\Big\{\alpha+\E\Big[\ell^\ast\Big(\lambda\frac{d\probq}{d\probp}\Big)\Big]\Big\}+\inf_{\vartheta\in\Theta}\E_\probq\Big[-\int^{T}_{0}\vartheta_t dS_t\Big]-\E_\probq[X]\bigg\}\,.
\end{eqnarray*}
Indeed, every functional $\psi\in(H^\Phi)'_+$ with $\psi(1_\Omega)=1$ can be identified with some $\probq\in\cP^\Phi$ such that $\psi(X')=\E_\probq[X']$ for all $X'\in H^\Phi$. The equivalent formulation for the term $\sigma_{\cA_\ell}$ follows from Theorem~10 in~\cite{FoellmerSchied2002}, whose proof extends easily to the space $H^\Phi$. The attainability of the supremum in the above representation is a consequence of the finiteness and continuity of $\rho_\ell$, see Remark~\ref{remark: attainability}.
\end{proof}

\medskip

\begin{remark}
\label{advantage geometrical approach with arai}
The above result shows that, even in the familiar context of cash-additive risk measures, an approach based on support functions provides an enhanced insight into dual representations. Indeed, the dual representation for $\rho_\ell$ established in Proposition~3.5 in~\cite{Arai2011} requires introducing the two auxiliary sets $\widetilde\cA$ and $\cA^1$ which make the structure of the corresponding dual representation unnecessarily involved compared to~\eqref{dual repr loss risk measure}. Our treatment seems to be more efficient mainly because we focus from the start on the structure of $\cA_\ell-\cC$ rather than on conjugate functions.
\end{remark}

%%%%%%%%%%%%%%%%%%%%%%%%%%%%%%%%%%%%%%%%%%%%%%%%%%

\subsection{Optimal risk sharing with multi-asset risk measures}
\label{inf conv section}

In this brief section we show how multi-asset risk measures arise naturally in the context of optimal risk sharing amongst several business lines. We focus on two business lines for ease of notation, the extension to a general number of business lines being straightforward.

\medskip

Throughout this section $\cX$ is a locally convex ordered topological vector space, $\cM$ the marketed space and $\pi:\cM\to\R$ the pricing functional. As above, we assume that the market is free of arbitrage by requiring $\pi$ to be a strictly positive functional. Assume two business lines have different types of capital requirements represented respectively by $\rho_{\cA,U,\pi}$ and $\rho_{\cB,V,\pi}$ where $\cA$ and $\cB$ are acceptance sets in $\cX$ and $U$ and $V$ are nonzero positive payoffs in $\cM$. We assume a ``risk'' $X\in\cX$ can be shared amongst the two business lines, i.e. for any $Y\in\cX$ we can assign $Y$ to the first and $X-Y$ to the second business line. The total required capital is then $\rho_{\cA,U,\pi}(Y)+\rho_{\cB,V,\pi}(X-Y)$. Optimal risk sharing is about finding the optimal $Y$. Hence, one naturally arrives at the \textit{infimal convolution} between $\rho_{\cA,U,\pi}$ and $\rho_{\cB,V,\pi}$ at $X$ given by
\begin{equation}
\rho_{\cA,U,\pi}\square\,\rho_{\cB,V,\pi}(X):=\inf\{\,\rho_{\cA,U,\pi}(Y)+\rho_{\cB,V,\pi}(X-Y): \, Y\in\cX\,\}\,.
\end{equation}
This quantity represents the ``minimum'' total required  capital across all possible allocations $(Y,X-Y)$ of the aggregated position $X$. For more details on the applications of infimal convolutions in the theory of risk measures we refer to Barrieu and El Karoui~\cite{BarrieuElKaroui2005}.

\medskip

First, we show that the infimal convolution of single-asset risk measures can be expressed as a multi-asset risk measure. Recall that a map $\rho:\cX\to\overline{\R}$ is said to be proper if it cannot assume the value $-\infty$ and its effective domain is nonempty.

\medskip

\begin{proposition}
\label{infimal theorem} Assume $\rho_{\cA,U,\pi}$ and $\rho_{\cB,V,\pi}$ are proper, and let $\elig\subset\cM$ be spanned by $U$ and $V$. Then for all $X\in\cX$
\begin{equation}
\label{aux inf}
\rho_{\cA,U,\pi}\square\,\rho_{\cB,V,\pi}(X)=\rho_{\cA+\cB,\elig,\pi}(X)\,.
\end{equation}
\end{proposition}
\begin{proof}
Take $X\in\cX$. To show that $\rho_{\cA,U,\pi}\square\,\rho_{\cB,V,\pi}(X)\leq\rho_{\cA+\cB,\elig,\pi}(X)$, assume $X+Z\in\cA+\cB$ for some $Z\in\elig$. Let $A\in\cA$ and $B\in\cB$ such that $X+Z=A+B$. Since $Z=\alpha U+\beta V$ for some $\alpha,\beta\in\R$, we obtain
\begin{equation}
\rho_{\cA,U,\pi}\square\,\rho_{\cB,V,\pi}(X)\leq\rho_{\cA,U,\pi}(A-\alpha U)+\rho_{\cB,V,\pi}(B-\beta V)\leq\pi(\alpha U)+\pi(\beta V)=\pi(Z)\;.
\end{equation}
The inequality follows by taking the infimum over all $Z\in\elig$ with $X+Z\in\cA+\cB$. To show the converse inequality, assume $Y+\alpha U\in\cA$ and $X-Y+\beta V\in\cB$ for some $Y\in\cX$ and $\alpha,\beta\in\R$. Then, clearly, $X+\alpha U+\beta V\in\cA+\cB$ and, therefore, $\rho_{\cA+\cB,\elig,\pi}(X)\leq\pi(\alpha U)+\pi(\beta V)$. Taking the infimum over all $\alpha$ and $\beta$ such $Y+\alpha U\in\cA$ and $X-Y+\beta V\in\cB$ yields $\rho_{\cA+\cB,\elig,\pi}(X)\le\rho_{\cA,U,\pi}(Y)+\rho_{\cB,V,\pi}(X-Y)$. We obtain the desired inequality after taking the infimum over all $Y\in\cX$.
\end{proof}

\medskip

The following corollary follows immediately from the preceding result and shows that every multi-asset risk measure with respect to a finite dimensional eligible space is in fact an infimal convolution of single-asset risk measures.

\medskip

\begin{corollary}
\label{infimal corollary}
Assume $\rho_{\cA,U,\pi}$ and $\rho_{\cA,V,\pi}$ are proper, and let $\elig$ be the span of $U$ and $V$. Then for every $X\in\cX$
\begin{equation}
\rho_{\cA,\elig,\pi}(X)=\rho_{\cA,U,\pi}\square\,\rho_{\cX_+,V,\pi}(X)\,.
\end{equation}
If $\cA$ is coherent, then for all $X\in\cX$
\begin{equation}
\rho_{\cA,\elig,\pi}(X)=\rho_{\cA,U,\pi}\square\,\rho_{\cA,V,\pi}(X)\,.
\end{equation}
\end{corollary}

\medskip

As a final result, we provide a dual representation for infimal convolutions of convex single-asset risk measures.

\medskip

\begin{proposition}
\label{dual repr inf conv of risk measures}
Assume $\cA$ and $\cB$ are convex with $\Interior(\cA)$ nonempty. Let $\rho_{\cA,U,\pi}$ and $\rho_{\cB,V,\pi}$ be proper. Moreover, assume the span $\elig$ of $U$ and $V$ contains a strictly positive element. If $\rho_{\cA+\cB,\elig,\pi}(0)>-\infty$, then for all $X\in\cX$
\begin{equation}
\label{dual repr inf conv}
\rho_{\cA,U,\pi}\square\,\rho_{\cB,V,\pi}(X)=\max_{\psi\in\extens}\{\sigma_{\cA}(\psi)+\sigma_{\cB}(\psi)-\psi(X)\}\,.
\end{equation}
In particular, if $\cA$ and $\cB$ are coherent then for all $X\in\cX$
\begin{equation}
\label{dual repr inf conv, coherent}
\rho_{\cA,U,\pi}\square\,\rho_{\cB,V,\pi}(X)=\max_{\psi\in\mathscr{D}}\psi(-X)
\end{equation}
where
\begin{equation}
\mathscr{D}:=\{\,\psi\in\extens: \, \psi(Y)\geq0, \;\, \forall \ Y\in\cA\cup\cB\,\}\,.
\end{equation}
\end{proposition}
\begin{proof}
By Proposition~\ref{infimal theorem} we know that $\rho_{\cA,U,\pi}\square\,\rho_{\cB,V,\pi}$ coincides with the multi-asset risk measure $\rho_{\cA+\cB,\elig,\pi}$. Note that $\cA+\cB$ is a convex acceptance set with nonempty interior. In particular, $\cA+\cB$ is a proper subset of $\cX$ since otherwise $\rho_{\cA+\cB,\elig,\pi}(0)=-\infty$. Then $\rho_{\cA,U,\pi}\square\,\rho_{\cB,V,\pi}$ is continuous by virtue of Proposition~\ref{convex strictly positive}. As a result, the representations~\eqref{dual repr inf conv} and~\eqref{dual repr inf conv, coherent} follow from Theorem~\ref{pointwise repr formula for multiple assets} and Corollary~\ref{representation of coherent}, respectively. In particular, the expression for the domain $\mathscr{\cD}$ is a consequence of $B(\cA+\cB)=B(\cA)\cap B(\cB)$. In both cases, the attainability of the supremum follows from Remark~\ref{remark: attainability}.
\end{proof}


\begin{thebibliography}{99}

\bibitem{AliprantisBorder2006} Aliprantis, Ch.D., Border, K.C.: \textit{Infinite Dimensional Analysis. A Hitchhiker's Guide}, 3rd edition, Springer (2006)

\bibitem{Arai2011} Arai, T.: Good deal bounds induced by shortfall risk, \textit{SIAM Journal of Financial Mathematics}, 2(1), 1-21 (2011)

\bibitem{ArtznerDelbaenEberHeath1999} Artzner, Ph., Delbaen, F., Eber, J.-M., Heath, D.: Coherent measures of risk, \textit{Mathematical Finance}, 9(3), 203-228 (1999)

\bibitem{ArtznerDelbaenKoch2009} Artzner, Ph., Delbaen, F., Koch-Medina, P.: Risk measures and efficient use of capital, \textit{ASTIN Bulletin}, 39(1), 101-116 (2009)

\bibitem{AubinEkeland1984} Aubin, J.-P., Ekeland, I.: \textit{Applied Nonlinear Analysis}, Dover Publications (2006)

\bibitem{BarrieuElKaroui2005} Barrieu, P., El Karoui, N.: Inf-convolution of risk measures and optimal risk transfer, \textit{Finance and Stochastics}, 9(2), 269-298 (2005)

\bibitem{Bauer1957} Bauer, H.: Sur le prolongement des formes lin\'{e}aires positives dans un espace vectoriel ordonn\'{e}, \textit{Comptes Rendus de l'Acad\'{e}mie des Sciences Paris}, 244, 289-292 (1957)

\bibitem{Clark1993} Clark, S.A.: The valuation problem in arbitrage price theory, \textit{Journal of Mathematical Economics}, 22(5), 463-478 (1993)

\bibitem{EdgarSucheston1992} Edgar,G.L., Sucheston L.: \textit{Stopping Times and Directed Processes}, Cambridge University Press (1992)

\bibitem{EkelandTemam1976} Ekeland I., T\'{e}mam, R.: \textit{Convex Analysis and Variational Problems}, SIAM (1999)

\bibitem{FarkasKochMunari2012a} Farkas, W., Koch-Medina P., Munari, C.: Beyond cash-additive risk measures: when changing the num\'{e}raire fails, {\em Finance and Stochastics}, 18(1), 145-173 (2014)

\bibitem{FarkasKochMunari2013b}
Farkas, W., Koch-Medina, P., Munari, C.: Capital requirements with defaultable securities, {\em Insurance: Mathematics and Economics}, 55, 58-67 (2014)

\bibitem{FoellmerSchied2002} F\"{o}llmer, H., Schied, A.: Convex measures of risk and trading constraints, \textit{Finance and Stochastics}, 6(4), 429-447 (2002)

\bibitem{FoellmerSchied2011} F\"{o}llmer, H., Schied, A.: \textit{Stochastic Finance: An Introduction in Discrete Time}, 3rd edition, de Gruyter (2011)

\bibitem{FrittelliScandolo2006} Frittelli, M., Scandolo, G.: Risk measures and capital requirements for processes, \textit{Mathematical Finance}, 16(4), 589-612 (2006)

\bibitem{HamelHeyde2010} Hamel, A.H., Heyde, F.: Duality for set-valued measures of risk, \textit{SIAM Journal on Financial Mathematics}, 1(1), 66-95 (2010)

\bibitem{HamelHeydeRudloff2011} Hamel, A.H., Heyde F., Rudloff, B.: Set-valued risk measures for conical market models, \textit{Mathematics and Financial Economics}, 5(1), 1-28 (2011)

\bibitem{Hustad1960} Hustad, O.: Linear inequalities and positive extension of linear functionals, \textit{Mathematica Scandinavica}, 8, 333-338 (1960)

\bibitem{JaschkeKuchler2001} Jaschke, S., K\"{u}chler, U.: Coherent risk measures and good deal bounds, \textit{Finance and Stochastics}, 5(2), 181-200 (2001)

\bibitem{Kountzakis2009} Kountzakis, C.E.: Generalized coherent risk measures, \textit{Applied Mathematical Sciences}, 3(49), 2437-2451 (2009)

\bibitem{Kreps1981} Kreps, D.M.: Arbitrage and equilibrium in economies with infinitely many commodities, \textit{Journal of Mathematical Economics}, 8(1), 15-35 (1981)

\bibitem{Namioka1957} Namioka, I.: Partially ordered linear topological spaces, \textit{Memoirs of the American Mathematical Society}, 24, Providence (1957)

\bibitem{Scandolo2004} Scandolo, G.: Models of capital requirements in static and dynamic settings, \textit{Economic Notes by Banca Monte dei Paschi di Siena SpA}, 33(3), 415-435 (2004)

\bibitem{Schachermayer2004} Schachermayer, W.: The fundamental theorem of asset pricing under proportional transaction costs in finite discrete time, \textit{Mathematical Finance}, 14(1), 19-48 (2004)

\bibitem{Zalinescu2002} Z\v{a}linescu, C.: \textit{Convex Analysis in General Vector Spaces}, World Scientific (2002)

\end{thebibliography}
\end{document}